\PassOptionsToPackage{square,comma,numbers, sort&compress}{natbib}
\documentclass[final, 5p, times, onecolumn]{elsarticle}

\usepackage{graphicx, color}
\usepackage{array, dsfont, mathrsfs, epstopdf}
\usepackage{amssymb}
\usepackage[cmex10]{amsmath}

\newtheorem{theorem}{Theorem}

\newtheorem{lemma}[theorem]{Lemma}
\newtheorem{proposition}[theorem]{Proposition}
\newdefinition{definition}{Definition}
\newdefinition{assumption}{Assumption}
\newdefinition{example}{Example}
\newdefinition{remark}{Remark}
\newproof{proof}{Proof}

\newcommand{\overbar}[1]{\mkern 1.5mu\overline{\mkern-1.5mu#1\mkern-1.5mu}\mkern 1.5mu}
\newcommand{\Id}{\mathrm{Id}}
\newcommand{\inte}{\mathrm{Int}}

\newcommand{\diag}{\mathrm{diag}}

\newcommand{\rzmk}{\mathsf{r}}
\newcommand{\ksvk}{\mathsf{k}}
\newcommand{\cclf}{\mathsf{c}}
\newcommand{\lyaR}{\mathsf{lr}}
\newcommand{\lyaK}{\mathsf{lk}}
\newcommand{\zbrf}{\mathsf{rz}}
\newcommand{\rbrf}{\mathsf{rr}}
\newcommand{\zbkf}{\mathsf{kz}}
\newcommand{\rbkf}{\mathsf{kr}}
\newcommand{\dly}{\text{d}}

\newcommand{\sgn}{\text{psgn}}
\newcommand{\sign}{\text{sgn}}
\newcommand{\master}{\text{m}}
\newcommand{\slave}{\text{s}}

\definecolor{bluencs}{rgb}{0.0, 0.53, 0.74}

\begin{document}

\begin{frontmatter}

\title{Razumikhin and Krasovskii Approaches for Safe Stabilization\tnoteref{label0}}
\tnotetext[label0]{This work was supported by the H2020 ERC Consolidator Grants L2C (864017) and LEAFHOUND (864720), the CHIST-ERA 2018 project DRUID-NET, the Walloon Region and the Innoviris Foundation, the Swedish Research Council (VR) and the Knut och Alice Wallenberg Foundation (KAW).}

\author[label1]{Wei Ren\corref{cor1}}\ead{w.ren@uclouvain.be}
\author[label1]{Rapha\"el M. Jungers}\ead{raphael.jungers@uclouvain.be}
\author[label2]{Dimos V. Dimarogonas} \ead{dimos@kth.se}

\cortext[cor1]{Corresponding author.}
\address[label1]{ICTEAM institute, Universit\'{e} catholique de Louvain, 1348 Louvain-la-Neuve, Belgium.}
\address[label2]{Division of Decision and Control Systems, EECS, KTH Royal Institute of Technology, SE-10044, Stockholm, Sweden.}

\begin{abstract}
This paper studies the stabilization and safety problems of nonlinear time-delay systems. Following both Razumikhin and Krasovskii approaches, we propose novel control Lyapunov functions/functionals for the stabilization problem and novel control barrier functions/functionals for the safety problem. The proposed control Lyapunov and barrier functions/functionals extend the existing ones from the delay-free case to the time-delay case, and allow for designing the stabilizing and safety controllers in closed-form. Since analytical solutions to time-delay optimal control problems are hard to be achieved, a sliding mode control based approach is developed to merge the proposed control Lyapunov and barrier functions/functionals. Based on the sliding surface functional, a feedback control law is established to investigate the stabilization and safety objectives simultaneously. In particular, the properties of the sliding surface functional are analyzed, and further how to construct the sliding surface functional is discussed. Finally, the proposed approaches are illustrated via two numerical examples from the connected cruise control problem of automotive systems and the synchronization problem of multi-agent systems.
\end{abstract}


\end{frontmatter}

\section{Introduction}
\label{sec-intro}

In physical domains like automotive, manufacturing, transportation and smart grids \citep{Lee2016introduction}, many physical systems are safety-critical. For these systems, two fundamental objectives are stabilization and safety \citep{Lamport1977proving}, which require physical systems to be operated stably and safely. In particular, the safety objective is placed in the priority position since it is imperative to keep a safety-critical system into a certain safe set while controlling it. For this purpose, the safe set is treated as a safety constraint involving system state and input, and the control design needs to comply with the safety constraint. Based on the classic stabilizing control through control Lyapunov functions (CLFs) \citep{Sontag1989universal, Jankovic2001control, Pepe2017control}, safety constraints can be specified in terms of a set invariance and verified via control barrier functions (CBFs), which are further used to design safety controllers \citep{Ogren2001control, Ngo2005integrator, Ames2016control, Romdlony2016stabilization, Maghenem2019characterizations}. CLFs and CBFs have been applied to deal with different objectives for diverse dynamical systems \citep{Ogren2001control, Panagou2015distributed, Lindemann2018control}. Especially, CLFs and CBFs have been combined to study stabilization and safety objectives simultaneously, and the combination can be either implicit \citep{Ngo2005integrator} or explicit \citep{Ames2016control, Jankovic2018robust, Romdlony2016stabilization} via different techniques.

In the fields of engineering, biology and physics, time delays are frequently encountered due to information acquisition and computation for control decisions and executions \citep{Sipahi2011stability}, and may induce many undesired phenomena like oscillation, instability and performance deterioration. However, since time delays cause a violation of monotonic decrease conditions \citep{Fridman2014tutorial, Zhou2016razumikhin}, classic Lyapunov theory cannot be applied directly to time-delay systems. There are two ways to extend the Lyapunov-based method \citep{Fridman2014tutorial, Zhou2016razumikhin, Ren2019krasovskii}. The first one is the Razumikhin approach \cite{Dashkovskiy2012lyapunov, Jankovic2001control}, and its essence is Lyapunov-Razumikhin functions, which are positive definite and whose derivative is negative definite under the Razumikhin condition. The second one is the Krasovskii approach based on Lyapunov-Krasovskii functionals \cite{Pepe2017control}, which are positive definite functionals with a negative definite derivative along the system solution. These two approaches have been used successfully in stability analysis and controller design of time-delay systems, such as functional differential systems \citep{Hale1993introduction}, hybrid systems \citep{Ren2019krasovskii, Ren2018vector} and multi-agent systems \citep{Papachristodoulou2010effects}. Although many physical systems are indeed safety-critical, to the best of our knowledge, there are few works \citep{Orosz2019safety, Prajna2005methods} on the safety of time-delay systems, which is the topic of this paper.

In order to investigate stabilization and safety of nonlinear time-delay systems, we follow Razumikhin and Krasovskii approaches and propose different CLFs and CBFs, which is the first contribution of this paper. In terms of the Razumikhin approach, since the Razumikhin condition constrains the decreasing of Lyapunov-Razumikhin functions, we propose an alternative control Lyapunov-Razumikhin function (CLRF) via the steepest descent feedback controller \citep{Pepe2014stabilization} to overcome the verification of the Razumikhin condition. In addition, two types of control barrier-Razumikhin functions (CBRFs) are proposed for the first time for the safety objective of time-delay systems. In terms of the Krasovskii approach, based on  smoothly separable functionals \citep{Liu2011input, Pepe2016stabilization}, control Lyapunov-Krasovskii functionals (CLKFs) and control barrier-Krasovskii functionals (CBKFs) are proposed for the stabilization and safety objectives, respectively. Furthermore, along these two directions, Razumikhin-type and Krasovskii-type small control properties (SCPs) are developed. Therefore, for time-delay systems, the safety can be verified via the CBRFs/CBKFs, and the stabilizing controllers are designed in closed-form via the proposed CLRF/CLKFs, which extend existing results \citep{Orosz2019safety, Prajna2005methods, Ames2016control} and Sontag's formula \citep{Sontag1989universal, Jankovic2018robust, Romdlony2016stabilization} into the time-delay case.

The second contribution of this paper is to address the stabilization and safety objectives simultaneously via the sliding mode control (SMC) approach, which is different from many existing works \citep{Jankovic2018robust, Ames2016control} based on optimization techniques. The motivation of the proposed approach lies in that (nonlinear) time-delay optimal control problems are difficult to be solved to obtain closed-form analytical solutions \citep{Wu2019new}. To begin with, the general sliding surface functional is constructed as the combination of the proposed control Lyapunov and barrier functions/functionals. Based on the associated projection operator \citep{Sira1999general}, the properties of the sliding surface functional are investigated, which extend the existing ones in \citep{Sira1999general} into the time-delay case and lays the foundation for the controller design. With these properties of the sliding surface functional, the existence and continuity of the controller is verified via designing the sliding mode controller in a closed form to guarantee the stabilization and safety objectives simultaneously. Second, using the proposed control Lyapunov and barrier functions/functionals, we present how to construct the sliding surface functional to achieve the desired objectives, which contributes to the extension of many existing methods like artificial potential field method \citep{Rimon1992exact}. Finally, we discuss how to attenuate or avoid side-effects of SMC via some approaches like higher-order sliding surface functionals \citep{Shtessel2014sliding, Oguchi2006sliding}. In conclusion, different from the optimization based controller design \citep{Jankovic2018robust, Ames2016control}, we develop a novel and general control design framework for simultaneous stabilization and safety of nonlinear time-delay systems based on control Lyapunov and barrier functions/functionals.

The remainder of this paper is organized below. Preliminaries are presented in Section \ref{sec-nonconsys}. Razumikhin-type CLFs and CBFs are proposed in Section \ref{sec-Razumikhintype}, whereas Krasovskii-type CLFs and CBFs are developed in Section \ref{sec-Krasovskiitype}. The SMC based approach to combine the proposed CLFs and CBFs is presented in Section \ref{sec-combinedfunction}. The derived results are demonstrated in Section \ref{sec-examples}. Conclusions and further research are stated in Section \ref{sec-conclusion}.

\section{Problem Formulation}
\label{sec-nonconsys}

We start with the notation throughout this paper. Let $\mathbb{R}:=(-\infty, +\infty); \mathbb{R}^{+}:=[0, +\infty); \mathbb{N}:=\{0, 1, \ldots\}$ and $\mathbb{N}^{+}:=\{1, 2, \ldots\}$. Given $a, b\in\mathbb{R}^{n}$, $(a,b):=[a^{\top}, b^{\top}]^{\top}$. $|\cdot|$ denotes the Euclidean norm of a real vector, or the induced Euclidean norm of a matrix. Given $\delta>0$ and $y\in\mathbb{R}^{n}$, an open ball centered at $y$ with radius of $\delta$ is defined as $\mathbf{B}(y, \delta):=\{x\in\mathbb{R}^{n}: |x-y|<\delta\}$; and $\mathbf{B}(\delta):=\mathbf{B}(0, \delta)$. Given $A, B\subset\mathbb{R}^{n}$ and $\delta>0$, $A-\mathbf{B}(\delta):=\{x\in A: \mathbf{B}(x, \delta)\subset A\}$ and $B\setminus A:=\{x\in\mathbb{R}^{n}: x\in B, x\notin A\}$. $\mathcal{C}(\mathbb{R}^{n}, \mathbb{R}^{p})$ denotes the set of all continuously differentiable functions mapping $\mathbb{R}^{n}$ to $\mathbb{R}^{p}$, and $\mathcal{C}([a, b], \mathbb{R}^{n})$ denotes the set of continuous functions mapping $[a, b]$ to $\mathbb{R}^{n}$, where $a\leq b$. $\alpha\circ\beta(v):=\alpha(\beta(v))$ for any $\alpha, \beta\in\mathcal{C}(\mathbb{R}^{n}, \mathbb{R}^{n})$. Given $x\in\mathcal{C}([-\Delta, +\infty), \mathbb{R}^{n})$, for any $t\in\mathbb{R}^{+}$, let $x_{t}$ be an element of $\mathcal{C}([-\Delta, 0], \mathbb{R}^{n})$ defined by $x_{t}(\theta):=x(t+\theta)$ with $\theta\in[-\Delta, 0]$. For $\phi\in\mathcal{C}([-\Delta, 0], \mathbb{R}^{n})$ with $\Delta>0$ fixed, we denote $\|\phi\|:=\sup_{\theta\in[-\Delta, 0]}|\phi(\theta)|$. For any $V\in\mathcal{C}(\mathbb{R}^{+}, \mathbb{R})$, its upper Dini derivative is $D^{+}V(t):=\limsup_{s\rightarrow0^{+}}\frac{V(t+s)-V(t)}{s}$. For any $h: \mathcal{C}([-\Delta, 0], \mathbb{R}^{n})\rightarrow\mathbb{R}^{+}$, its upper Dini derivative is $D^{+}h(x_{t})=\limsup_{v\rightarrow0^{+}}\frac{h(x_{t+v})-h(x_{t})}{v}$. $\Id$ is the identity function, $\sign: \mathbb{R}\rightarrow\{-1, 0, 1\}$ is the sign function, and $\sgn: \mathbb{R}\rightarrow\{0, 1\}$ is defined by $\sgn(v)=1$ if $v\geq0$; $\sgn(v)=0$ otherwise. Given $\alpha\in\mathcal{C}(\mathbb{R}, \mathbb{R})$, $\alpha\geq\Id$ means that for all $v\in\mathbb{R}$, $\alpha(v)\geq v$. A function $\alpha: \mathbb{R}^{+}\rightarrow\mathbb{R}^{+}$ is of class $\mathcal{K}$ if it is continuous, $\alpha(0)=0$, and strictly increasing; it is of class $\mathcal{K}_{\infty}$ if it is of class $\mathcal{K}$ and unbounded. A continuous function $\alpha: \mathbb{R}\rightarrow\mathbb{R}$ is of class $\overbar{\mathcal{K}}$ if it is strictly increasing and $\alpha(0)=0$. A function $\beta: \mathbb{R}^{+}\times\mathbb{R}^{+}\rightarrow\mathbb{R}^{+}$ is of class $\mathcal{KL}$ if $\beta(s, t)$ is of class $\mathcal{K}$ for any fixed $t\geq0$ and $\beta(s, t)\rightarrow0$ as $t\rightarrow\infty$ for any fixed $s\geq0$.

In this paper, we focus on nonlinear time-delay control systems with the following dynamics:
\begin{align}
\label{eqn-1}
\begin{aligned}
\dot{x}(t)&=f(x_{t})+g(x_{t})u, \quad t>0, \\
x(t)&=\xi(t), \quad t\in[-\Delta, 0],
\end{aligned}
\end{align}
where $x\in\mathbb{R}^{n}$ is the system state, $x_{t}(\theta)=x(t+\theta)\in\mathbb{R}^{n}$ is the time-delay state with $\theta\in[-\Delta, 0]$, where $\Delta>0$ is the upper bound of time delays. The initial state is $\xi\in\mathcal{C}([-\Delta, 0], \mathbb{X}_{0})$ with $\mathbb{X}_{0}\subset\mathbb{R}^{n}$ and $\|\xi\|$ being bounded. The control input $u$ takes value from the set $\mathbb{U}\subseteq\mathbb{R}^{m}$. Here, the input function is not specified explicitly due to the allowance for its dependence on the current state only (i.e., $u: \mathbb{R}^{+}\rightarrow\mathbb{U}$) or the time-delay trajectory (i.e., $u: \mathcal{C}([-\Delta, 0], \mathbb{R}^{n})\rightarrow\mathbb{U}$). The functionals $f: \mathcal{C}([-\Delta, 0], \mathbb{R}^{n})\rightarrow\mathbb{R}^{n}$ and $g: \mathcal{C}([-\Delta, 0], \mathbb{R}^{n})\rightarrow\mathbb{R}^{n\times m}$ are assumed to be continuous and locally Lipschitz, which guarantees the existence of the unique solution to the system \eqref{eqn-1}; see \cite[Section 2]{Hale1993introduction}. Let $f(0)=0$. To consider the stabilization problem of the system \eqref{eqn-1}, we assume that the origin is included in the initial set $\mathbb{X}_{0}$.

\begin{definition}[\citep{Sastry2013nonlinear}]
\label{def-1}
Given a control input $u\in\mathbb{R}^{m}$, the system \eqref{eqn-1} is \emph{globally asymptotically stable (GAS)}, if there exists $\beta\in\mathcal{KL}$ such that $|x(t)|\leq\beta(\|\xi\|, t)$ for all $t\geq0$ and all bounded $\xi\in\mathcal{C}([-\Delta, 0], \mathbb{R}^{n})$; the system \eqref{eqn-1} is \emph{semi-globally asymptotically stable (semi-GAS)}, if there exists $\beta\in\mathcal{KL}$ such that $|x(t)|\leq\beta(\|\xi\|, t)$ for all $t\geq0$ and all $\xi\in\mathcal{C}([-\Delta, 0], \mathbb{X}_{0})$.
\end{definition}

From Definition \ref{def-1}, the stabilization objective is to design a feedback controller to guarantee the system \eqref{eqn-1} to be GAS. Note that if only the stabilization objective is studied, then we focus on the GAS property; otherwise, we focus on the semi-GAS property due to the safety objective. For the system \eqref{eqn-1}, a set $\mathbb{C}\subset\mathbb{R}^{n}$ is \emph{forward invariant}, if $x(t)\in\mathbb{C}$ for all $t\geq-\Delta$ and any trajectory $x(t)$ starting from an initial state $\xi\in\mathcal{C}([-\Delta, 0], \mathbb{X}_{0})$ with $\mathbb{X}_{0}\subseteq\mathbb{C}$. The system \eqref{eqn-1} is \emph{safe} with respect to the set $\mathbb{C}$, if the set $\mathbb{C}$ is forward invariant, and then the set $\mathbb{C}$ is called the \emph{safe set}. Hence, the safety objective is to design a feedback controller to keep the trajectory of the system \eqref{eqn-1} within the safe set. Since the safety and stabilization objectives of time-delay systems cannot be achieved via classic CLFs and CBFs \citep{Liu2016new}, the goal of this paper is to implement both Razumikhin and Krasovskii approaches to propose novel types of CLFs and CBFs for the system \eqref{eqn-1}.

\section{Razumikhin-type Control Functions}
\label{sec-Razumikhintype}

In this section, we follow the Razumikhin approach to propose control Lyapunov and barrier functions for time-delay systems. For this purpose, we first propose control Lyapunov-Razumikhin functions for the stabilization objective, and then control barrier-Razumikhin functions for the safety objective.

\subsection{Control Lyapunov-Razumikhin Functions}
\label{subsec-clrf}

We first recall the classic control Lyapunov-Razumikhin function.

\begin{definition}[\citep{Jankovic2001control, Pepe2017control}]
\label{def-2}
For the system \eqref{eqn-1}, a function $V_{\cclf}\in\mathcal{C}(\mathbb{R}^{n}, \mathbb{R}^{+})$ is called a \emph{control Lyapunov-Razumikhin function (CLRF-I)}, if
\begin{enumerate}[(i)]
  \item there exist $\alpha_{1}, \alpha_{2}\in\mathcal{K}_{\infty}$ such that, for all $x\in\mathbb{R}^{n}$, $\alpha_{1}(|x|)\leq V_{\cclf}(x)\leq\alpha_{2}(|x|)$;
  \item there exist $\gamma_{\cclf}, \rho\in\mathcal{K}$ with $\rho>\Id$ such that, for all $\phi\in\mathcal{C}([-\Delta, 0], \mathbb{R}^{n})$ with $\phi(0)=x$, if $\rho(V_{\cclf}(x))\geq V_{\cclf}(\phi(\theta))$ for all $\theta\in[-\Delta, 0]$, then
  \begin{align}
  \label{eqn-2}
  &\inf_{u\in\mathbb{U}}\left\{L_{f}V_{\cclf}(\phi)+L_{g}V_{\cclf}(\phi)u\right\}\leq-\gamma_{\cclf}(V_{\cclf}(x)),
  \end{align}
\end{enumerate}
where $L_{f}V_{\cclf}(\phi):=\frac{\partial V_{\cclf}(x)}{\partial x}f(\phi)$ and $L_{g}V_{\cclf}(\phi):=\frac{\partial V_{\cclf}(x)}{\partial x}g(\phi)$.
\end{definition}

In Definition \ref{def-2}, the subscript `$\cclf$' is to show that the CLRF is `classic', which is to distinguish with the CLRF proposed afterwards. In \eqref{eqn-2}, $L_{f}V_{\cclf}(\phi)+L_{g}V_{\cclf}(\phi)u$ is an alternative form of the upper Dini derivative of $V_{\cclf}$, and here the upper Dini derivative equals to the usual derivative. The effects of time delays are shown via the Razumikhin condition: $\rho(V_{\cclf}(x))\geq V_{\cclf}(\phi(\theta))$ for all $\theta\in[-\Delta, 0]$, which can be written equivalently as $\rho(V_{\cclf}(x))\geq\|V_{\cclf}(\phi)\|$; see \citep{Pepe2017control, Ren2018vector}. If $L_{g}V_{\cclf}(\phi)=0$, then Definition \ref{def-2} is reduced to the one in \citep{Jankovic2001control}.

In the delay-free case, the existence of the continuous controller is verified via the SCP \citep{Sontag1989universal}. However, in the time-delay case, whether the Razumikhin condition is satisfied results in additional difficulties in the controller design, and the SCP is unavailable. To derive a stabilizing controller, the condition on the CLRF-I is presented in the following theorem, which is an extension of Theorem 1 in \citep{Jankovic2001control}, and a simplified proof is given.

\begin{theorem}
\label{thm-1}
For the system \eqref{eqn-1} with a CLRF-I $V_{\cclf}\in\mathcal{C}(\mathbb{R}^{n}, \mathbb{R}^{+})$, if there exists $\chi\in\mathbb{R}^{+}$ such that for all $\phi\in\mathcal{C}([-\Delta, 0], \mathbf{B}(\varepsilon))$ with some $\varepsilon>0$ and $L_{g}V_{\cclf}(\phi)\neq0$,
\begin{equation}
\label{eqn-3}
L_{f}V_{\cclf}(\phi)\leq\chi\|L_{g}V_{\cclf}(\phi)\|^{2},
\end{equation}
then there exists a smooth function $\varphi: \mathbb{R}^{+}\rightarrow\mathbb{R}^{+}$ satisfying $\int^{\infty}_{0}\varphi(v)dv=\infty$ such that
\begin{align}
\label{eqn-4}
u(\phi):=-\varphi(V_{\cclf}(x))(L_{g}V_{\cclf}(\phi))^{\top}
\end{align}
 is a continuous controller for the system \eqref{eqn-1} and ensures the closed-loop system to be GAS.
\end{theorem}

\begin{proof}
We first show the convergence of $V_{\cclf}(x)$ to a region around the origin, and then ensure the convergence of $V_{\cclf}(x)$ in this region. In this way, the system stability can be concluded from these two parts.

\emph{\textbf{Part 1:} convergence of $V_{\cclf}(x)$ to a region around the origin.} We consider two cases: $L_{f}V_{\cclf}(\phi)\leq-\varepsilon\gamma_{\cclf}(V_{\cclf}(x))$ and $L_{f}V_{\cclf}(\phi)\geq-\varepsilon\gamma_{\cclf}(V_{\cclf}(x))$, where $\varepsilon\in(0, 1)$.

For the first case, we have from \eqref{eqn-4} that
\begin{align}
\label{eqn-5}
L_{f}V_{\cclf}(\phi)+L_{g}V_{\cclf}(\phi)u&=L_{f}V_{\cclf}(\phi)-\varphi(V_{\cclf}(x))\|L_{g}V_{\cclf}(\phi)\|^{2} \nonumber \\
&\leq-\varepsilon\gamma_{\cclf}(V_{\cclf}(x))\leq0,
\end{align}
where the first ``$\leq$'' holds due to the positivity of the function $\varphi$.

For the second case, from Lemma 1 in \citep{Jankovic2001control}, $L_{g}V_{\cclf}(\phi)\neq0$ and there exists $\omega>0$ such that $\|L_{g}V_{\cclf}(\phi)\|>\omega$. Define the function $\varphi_{1}: \mathbb{R}^{+}\rightarrow\mathbb{R}^{+}$:
\begin{align}
\label{eqn-6}
\varphi_{1}(V_{\cclf}(x)):=\sup_{\theta\in[-\Delta, 0]}\frac{\varepsilon+\mathfrak{a}_{\cclf}(\phi(\theta))}{\|\mathfrak{b}_{\cclf}(\phi)\|^{2}+\varpi(\mathfrak{b}_{\cclf}(\phi))},
\end{align}
where $\mathfrak{a}_{\cclf}(\phi):=L_{f}V_{\cclf}(\phi)+\gamma_{\cclf}(V_{\cclf}(x))$, $\mathfrak{b}_{\cclf}(\phi):=L_{g}V_{\cclf}(\phi)$, and $\varpi(\mathfrak{b}_{\cclf}(\phi))=\epsilon(1-\sign(\|\mathfrak{b}_{\cclf}(\phi)\|))$ with sufficiently small $\epsilon>0$. Since $L_{f}V_{\cclf}(\phi)\geq-\varepsilon\gamma_{\cclf}(V_{\cclf}(x))$, $\mathfrak{a}_{\cclf}(\phi)\geq(1-\varepsilon)\gamma_{\cclf}(V_{\cclf}(x))>0$. Hence, $\varphi_{1}(V_{\cclf}(\phi))>0$ and $\int^{\infty}_{0}\varphi_{1}(v)dv=\infty$. From $V_{\cclf}\in\mathcal{C}(\mathbb{R}^{n}, \mathbb{R}^{+})$ and the continuity of $\mathfrak{a}_{\cclf}(\phi)$ and $\mathfrak{b}_{\cclf}(\phi)$, we can deduce the continuity of the function $\varphi_{1}(V_{\cclf}(x))$. From \eqref{eqn-2} and \eqref{eqn-6},
\begin{align}
\label{eqn-7}
L_{f}V_{\cclf}(\phi)-\varphi_{1}(V_{\cclf}(x))\|L_{g}V_{\cclf}(\phi)\|^{2}&\leq L_{f}V_{\cclf}(\phi)-\varepsilon-\mathfrak{a}_{\cclf}(\phi)  \nonumber \\
&=-\varepsilon-\gamma_{\cclf}(V_{\cclf}(x))  \nonumber \\
&<-\gamma_{\cclf}(V_{\cclf}(x)).
\end{align}
From \eqref{eqn-5}-\eqref{eqn-7} and the Razumikhin theorem in \citep{Hale1993introduction}, we conclude that, under the controller \eqref{eqn-4} with $\varphi(v)\geq\varphi_{1}(v)$ for all $v\in\mathbb{R}^{+}$, $V_{\cclf}(x)$ converges to a region around the origin. Let this region be $\mathbf{B}_{\textsf{v}}(\mathfrak{p}):=\{x\in\mathbb{R}^{n}: V_{\cclf}(x)\leq\mathfrak{p}\}$ with $\mathfrak{p}>0$.

\emph{\textbf{Part 2:} asymptotic stability of $V_{\cclf}(x)$ in} $\mathbf{B}_{\textsf{v}}(\mathfrak{p})$. We consider the following two cases.
\begin{enumerate}[{i)}]
  \item If $L_{g}V_{\cclf}(\phi)=0$, then $L_{f}V_{\cclf}(\phi)\leq-\gamma_{\cclf}(V_{\cclf}(x))$ from \eqref{eqn-2}.

  \item If $L_{g}V_{\cclf}(\phi)\neq0$, then \eqref{eqn-3} holds. If $L_{f}V_{\cclf}(\phi)\leq-\varepsilon\gamma_{\cclf}(V_{\cclf}(x))$, then $L_{f}V_{\cclf}(\phi)-\varphi(V_{\cclf}(x))\|L_{g}V_{\cclf}(\phi)\|^{2}<-\varepsilon\gamma_{\cclf}(V_{\cclf}(x))$ from \eqref{eqn-6}. If $L_{f}V_{\cclf}(\phi)\geq-\varepsilon\gamma_{\cclf}(V_{\cclf}(x))$, then from Lemma 1 in \citep{Jankovic2001control}, there exists $\omega_{1}>0$ such that $\|L_{g}V_{\cclf}(\phi)\|>\omega_{1}$. Note that $\omega_{1}$ depends on $|x|$, since we focus on the region $\mathbf{B}_{\textsf{v}}(\mathfrak{p})$ here. Let $\varphi_{1}(V_{\cclf}(x)):=1+\chi$. Thus, $L_{f}V_{\cclf}(\phi)-\varphi_{1}(V_{\cclf}(x))\|L_{g}V_{\cclf}(\phi)\|^{2}\leq-\|L_{g}V_{\cclf}(\phi)\|^{2}<-\omega^{2}_{1}$.
\end{enumerate}
Therefore, $\dot{V}_{\cclf}(x)<0$ for all $x\in\mathbf{B}_{\textsf{v}}(\mathfrak{p})\setminus\{0\}$, and the asymptotic stability in $\mathbf{B}_{\textsf{v}}(\mathfrak{p})$ is guaranteed via the Razumikhin theorem.

Finally, based on the function $\varphi_{1}$, the function $\varphi$ can be constructed to be smooth and to satisfy: $\varphi(v)\geq\varphi_{1}(v)$ for all $v\geq\mathfrak{p}$; $\varphi(v)\geq1+\chi$ for all $v\leq\mathfrak{p}$. Therefore, the above analysis is still valid with the function $\varphi$, which is continuous and further ensures the continuity of the controller \eqref{eqn-4}.
\hfill$\blacksquare$
\end{proof}

Theorem \ref{thm-1} shows the existence of the continuous controller for the system \eqref{eqn-1} under the condition \eqref{eqn-3}. The condition \eqref{eqn-3} is similar to the SCP to guarantee the continuity of the derived controller \eqref{eqn-4} at the origin, and further implies that the derivative of the CLRF-I $V_{\cclf}$ is negative in a neighborhood of the origin \cite[Chapter 3.4]{Sepulchre2012constructive}. The controller \eqref{eqn-4} is called the \emph{domination redesign control law}, which is derived based on the optimal stabilization problem \citep{Sepulchre2012constructive}. The corresponding function $\varphi$ is called the \emph{dominating function}, and always exists because $V_{\cclf}$ is differentiable and radially unbounded \citep{Jankovic2001control, Sepulchre2012constructive}.

From Theorem \ref{thm-1}, the condition \eqref{eqn-3} needs to be verified in the controller design, however, the verification of the condition \eqref{eqn-3} is not easy, in particular, for high-dimensional systems. On the other hand, since the function $\varphi$ in \eqref{eqn-4} is related to the optimization problem \citep{Sepulchre2012constructive} or is constructed based on the sublevel set of the CLRF-I \citep{Jankovic2001control}, the closed form of the controller cannot be expressed easily and explicitly. To avoid the verification of the Razumikhin condition and to establish the controller explicitly, we propose an alternative CLRF, which is based on the steepest descent feedback controller \citep{Clarke2010discontinuous, Pepe2017control}.

\begin{definition}
\label{def-3}
For the system \eqref{eqn-1}, a function $V_{\rzmk}\in\mathcal{C}(\mathbb{R}^{n}, \mathbb{R}^{+})$ is called a \emph{control Lyapunov-Razumikhin function (CLRF-II)}, if item (i) in Definition \ref{def-2} holds, and there exist $\gamma_{\rzmk}, \eta_{\rzmk}\in\mathcal{K}$ such that $\gamma_{\rzmk}-\eta_{\rzmk}\in\mathcal{K}$, and for any nonzero $\phi\in\mathcal{C}([-\Delta, 0], \mathbb{R}^{n})$ with $\phi(0)=x$,
\begin{align}
\label{eqn-8}
\inf_{u\in\mathbb{U}}\left\{L_{f}V_{\rzmk}(\phi)+L_{g}V_{\rzmk}(\phi)u\right\}<-\gamma_{\rzmk}(V_{\rzmk}(x))+\eta_{\rzmk}(\|V_{\rzmk}(\phi)\|).
\end{align}
\end{definition}

Comparing with Definition \ref{def-2} where the Razumikhin condition is the premise of the condition \eqref{eqn-2}, the Razumikhin condition is not included in Definition \ref{def-3}. Hence, \eqref{eqn-8} can be treated as an extension of the condition \eqref{eqn-2}, in the sense that the item $\eta_{\rzmk}$ in \eqref{eqn-8} is to estimate the derivative of $V_{\rzmk}$ when the Razumikhin condition is not satisfied. If $g(0)=0$, then \eqref{eqn-8} is reduced to a non-strict inequality for all $\phi\in\mathcal{C}([-\Delta, 0], \mathbb{R}^{n})$. The relation between the CLRF-I and CLRF-II is shown in the following proposition, which is similar to Remark 4 in \cite{Pepe2017control}.

\begin{proposition}
\label{thm-2}
Consider the system \eqref{eqn-1}. If $V_{\rzmk}\in\mathcal{C}(\mathbb{R}^{n}, \mathbb{R}^{+})$ is a CLRF-II, then $V_{\rzmk}$ is a CLRF-I.
\end{proposition}

\begin{proof}
Since the CLRF-II $V_{\rzmk}$ satisfies item (i) in Definition \ref{def-2}, we only need to show that \eqref{eqn-2} can be derived from \eqref{eqn-8}. From the definition of the CLRF-II, there exist $\gamma_{\rzmk}, \eta_{\rzmk}\in\mathcal{K}$ such that $\gamma_{\rzmk}-\eta_{\rzmk}\in\mathcal{K}$, and \eqref{eqn-8} holds for any nonzero $\phi\in\mathcal{C}([-\Delta, 0], \mathbb{R}^{n})$. From Lemma A.1 in \citep{Pepe2021nonlinear}, there exists $\rho\in\mathcal{K}$ such that $\rho(s)>s$ for all $s>0$ and $\gamma_{\rzmk}-\eta_{\rzmk}\circ\rho\in\mathcal{K}$. In this respect, we have
\begin{align*}
&\inf_{u\in\mathbb{U}}\left\{L_{f}V_{\rzmk}(\phi)+L_{g}V_{\rzmk}(\phi)u\right\}<-\gamma_{\rzmk}(V_{\rzmk}(x))+\eta_{\rzmk}(\|V_{\rzmk}(\phi)\|) \\
&=-(\gamma_{\rzmk}-\eta_{\rzmk}\circ\rho)(V_{\rzmk}(x))-\eta_{\rzmk}(\rho(V_{\rzmk}(x)))+\eta_{\rzmk}(\|V_{\rzmk}(\phi)\|).
\end{align*}
If $\rho(V_{\rzmk}(x))\geq V_{\rzmk}(\phi(\theta))$ for all $\theta\in[-\Delta, 0]$, then
\begin{align*}
\inf_{u\in\mathbb{U}}\left\{L_{f}V_{\rzmk}(\phi)+L_{g}V_{\rzmk}(\phi)u\right\}&<-(\gamma_{\rzmk}-\eta_{\rzmk}\circ\rho)(V_{\rzmk}(x)).
\end{align*}
which implies that item (ii) of Definition \ref{def-2} holds. Hence, $V_{\rzmk}$ is a CLRF-I with $\gamma_{\cclf}:=\gamma_{\rzmk}-\eta_{\rzmk}\circ\rho$.
\hfill$\blacksquare$
\end{proof}

We emphasize that the CLRF-II $V_{\rzmk}$ cannot be deduced from the CLRF-I $V_{\cclf}$. The reason lies in the case where the Razumikhin condition is not satisfied. To be specific, without the Razumikhin condition, the derivative of $V_{\rzmk}$ is still bounded via \eqref{eqn-8}, whereas the bound of the derivative of $V_{\cclf}$ is unknown and cannot be derived from \eqref{eqn-2} directly. Therefore, the estimate of the CLRF-II is bounded consistently via \eqref{eqn-8}, which is not the case for the CLRF-I. In this respect, it is not easy to apply the CLRF-I directly to design the stabilizing controller explicitly (see, e.g., Theorem \ref{thm-1} and \citep{Jankovic2001control, Pepe2017control}). However, with the CLRF-II, we can follow Sontag's formula to design an explicit controller, which will be further discussed later.

From Definition \ref{def-3}, we define the following set
\begin{align}
\label{eqn-9}
\mathbb{K}_{\lyaR}&:=\{u\in\mathbb{U}: L_{f}V_{\rzmk}(\phi)+L_{g}V_{\rzmk}(\phi)u\leq-\gamma_{\rzmk}(V_{\rzmk}(x))+\eta_{\rzmk}(\|V_{\rzmk}(\phi)\|)\}.
\end{align}
The non-emptiness of the set $\mathbb{K}_{\lyaR}$ follows from the existence of the CLF-II $V_{\rzmk}$. To show this, if $\mathbb{K}_{\lyaR}=\varnothing$, then $L_{f}V_{\rzmk}(\phi)+L_{g}V_{\rzmk}(\phi)u>-\gamma_{\rzmk}(V_{\rzmk}(x))+\eta_{\rzmk}(\|V_{\rzmk}(\phi)\|)$ for all $u\in\mathbb{U}$, and $\inf_{u\in\mathbb{U}}\{L_{f}V_{\rzmk}(\phi)+L_{g}V_{\rzmk}(\phi)u\}>-\gamma_{\rzmk}(V_{\rzmk}(x))+\eta_{\rzmk}(\|V_{\rzmk}(\phi)\|)$, which contradicts with \eqref{eqn-8}. Hence, the set $\mathbb{K}_{\lyaR}$ is non-empty. Based on \citep[Corollary 1]{Pepe2021nonlinear}, any locally Lipschitz controller $u: \mathcal{C}([-\Delta, 0], \mathbb{R}^{n})\rightarrow\mathbb{U}$ with $u(\phi)\in\mathbb{K}_{\lyaR}$ results in the stabilization of the system \eqref{eqn-1}, which further shows the satisfaction of the stabilization objective via the CLRF-II. To derive the controller explicitly, we propose the Razumikhin-type SCP for the system \eqref{eqn-1} based on the CLRF-II.

\begin{definition}
\label{def-4}
Consider the system \eqref{eqn-1} with a CLRF-II $V_{\rzmk}\in\mathcal{C}(\mathbb{R}^{n}, \mathbb{R}^{+})$, the system \eqref{eqn-1} is said to satisfy the \emph{Razumikhin-type small control property (R-SCP)}, if for arbitrary $\varepsilon>0$, there exists $\delta>0$ such that for any nonzero $\phi\in\mathcal{C}([-\Delta, 0], \mathbf{B}(\delta))$ with $\phi(0)=x$, there exists $u\in\mathbf{B}(\varepsilon)$ such that
\begin{align}
\label{eqn-10}
L_{f}V_{\rzmk}(\phi)+L_{g}V_{\rzmk}(\phi)u<-\gamma_{\rzmk}(V_{\rzmk}(x))+\eta_{\rzmk}(\|V_{\rzmk}(\phi)\|).
\end{align}
\end{definition}

Based on the Razumikhin approach, Definition \ref{def-4} is a reformulation of the classic SCP for time-delay systems. With the CLRF-II and the R-SCP, the controller can be derived explicitly such that the closed-loop system \eqref{eqn-1} is GAS.

\begin{theorem}
\label{thm-3}
If the system \eqref{eqn-1} admits a CLRF-II $V_{\rzmk}\in\mathcal{C}(\mathbb{R}^{n}, \mathbb{R}^{+})$ and satisfies the R-SCP, then the controller defined by
\begin{align}
\label{eqn-11}
u(\phi):=\left\{\begin{aligned}
&\kappa(\lambda, \mathfrak{a}_{\rzmk}(\phi), (L_{g}V_{\rzmk}(\phi))^{\top}), &\text{ if }& \phi\neq 0, \\
&0, &\text{ if }&  \phi=0,
\end{aligned}\right.
\end{align}
with $\lambda>0$, $\mathfrak{a}_{\rzmk}(\phi):=L_{f}V_{\rzmk}(\phi)+\gamma_{\rzmk}(V_{\rzmk}(x))-\eta_{\rzmk}(\|V_{\rzmk}(\phi)\|)$, and
\begin{equation}
\label{eqn-12}
\kappa(\lambda, p, q)=\left\{\begin{aligned}
&\frac{p+\sqrt{p^{2}+\lambda\|q\|^{4}}}{-\|q\|^{2}}q, &\text{ if }\ & q\neq 0, \\
&0, &\text{ if }\ &  q=0,
\end{aligned}\right.
\end{equation}
is continuous and ensures the system \eqref{eqn-1} to be GAS.
\end{theorem}

\begin{proof}
We first prove the stabilization of the system \eqref{eqn-1} under the controller \eqref{eqn-11}. From \eqref{eqn-11}, if $L_{g}V_{\rzmk}(\phi)\equiv0$, then $u\equiv0$ and from \eqref{eqn-8}, we have
\begin{align*}
&L_{f}V_{\rzmk}(\phi)+L_{g}V_{\rzmk}(\phi)u+\gamma_{\rzmk}(V_{\rzmk}(x))-\eta_{\rzmk}(\|V_{\rzmk}(\phi)\|) \\
&=L_{f}V_{\rzmk}(\phi)+\gamma_{\rzmk}(V_{\rzmk}(x))-\eta_{\rzmk}(\|V_{\rzmk}(\phi)\|)<0.
\end{align*}
For the case $L_{g}V_{\rzmk}(\phi)\neq0$, we define $\mathfrak{a}_{\rzmk}(\phi):=L_{f}V_{\rzmk}(\phi)+\gamma_{\rzmk}(V_{\rzmk}(x))$ and $\mathfrak{b}_{\rzmk}(\phi):=L_{g}V_{\rzmk}(\phi)$. From \eqref{eqn-11}, we have
\begin{align*}
&L_{f}V_{\rzmk}(\phi)+L_{g}V_{\rzmk}(\phi)u+\gamma_{\rzmk}(V_{\rzmk}(x))-\eta_{\rzmk}(\|V_{\rzmk}(\phi)\|) \\
&=\mathfrak{a}_{\rzmk}(\phi)-\mathfrak{b}_{\rzmk}(\phi)\frac{\mathfrak{a}_{\rzmk}(\phi)\mathfrak{b}^{\top}_{\rzmk}(\phi)}{\|\mathfrak{b}_{\rzmk}(\phi)\|^{2}}
-\mathfrak{b}_{\rzmk}(\phi)\frac{\sqrt{\mathfrak{a}^{2}_{\rzmk}(\phi)+\lambda\|\mathfrak{b}_{\rzmk}(\phi)\|^{4}}}{\|\mathfrak{b}_{\rzmk}(\phi)\|^{2}}\mathfrak{b}^{\top}_{\rzmk}(\phi) \\
&=-\sqrt{\mathfrak{a}^{2}_{\rzmk}(\phi)+\lambda\|\mathfrak{b}_{\rzmk}(\phi)\|^{4}}<0,
\end{align*}
and thus $L_{f}V_{\rzmk}(\phi)+L_{g}V_{\rzmk}(\phi)u<-\gamma_{\rzmk}(V_{\rzmk}(x(t)))+\eta_{\rzmk}(\|V_{\rzmk}(\phi)\|)$.

Summarizing the above analysis, we have
\begin{align}
\label{eqn-13}
\dot{V}_{\rzmk}(x(t))<-\gamma_{\rzmk}(V_{\rzmk}(x(t)))+\eta_{\rzmk}(\|V_{\rzmk}(\phi)\|), \quad \forall t\in\mathbb{R}^{+}.
\end{align}
From Corollary 1 in \citep{Pepe2021nonlinear} and item (i) in Definition \ref{def-2}, there exists $\beta\in\mathcal{KL}$ such that $|x(t)|\leq\beta(\|\xi\|, t)$ for all $t\in\mathbb{R}^{+}$, which implies that the system \eqref{eqn-1} is GAS.

Next, we show the continuity of the controller. From $V_{\rzmk}\in\mathcal{C}(\mathbb{R}^{n}, \mathbb{R}^{+})$ and the continuity of $\mathfrak{a}_{\rzmk}(\phi)$ and $\mathfrak{b}_{\rzmk}(\phi)$, we can deduce the continuity of the controller in any region away from the origin. In the following, we only need to show the continuity of the controller at the origin. From the R-SCP, for arbitrary $\varepsilon>0$, there exists $\delta_{1}>0$ such that for nonzero $\phi\in\mathcal{C}([-\Delta, 0], \mathbf{B}(\delta_{1}))$ with $\phi(0)=x$, there exists $u\in\mathbf{B}(\varepsilon)$ such that
\begin{align}
\label{eqn-14}
L_{f}V_{\rzmk}(\phi)+L_{g}V_{\rzmk}(\phi)u&<-\gamma_{\rzmk}(V_{\rzmk}(x))+\eta_{\rzmk}(\|V_{\rzmk}(\phi)\|).
\end{align}
In addition, since $V_{\rzmk}$ is continuously differentiable and $g$ in \eqref{eqn-1} is locally Lipschitz, there exists $\delta_{2}\in\mathbb{R}^{+}$ with $\delta_{2}\neq\delta_{1}$ such that for nonzero $\phi\in\mathcal{C}([-\Delta, 0], \mathbf{B}(\delta_{2}))$,
\begin{align}
\label{eqn-15}
\|L_{g}V_{\rzmk}(\phi)\|&\leq\varepsilon.
\end{align}
Let $\delta:=\min\{\delta_{1}, \delta_{2}\}$, and then \eqref{eqn-14}-\eqref{eqn-15} hold for any nonzero $\phi\in\mathcal{C}([-\Delta, 0], \mathbf{B}(\delta))$.

From \eqref{eqn-11}, $u\equiv0$ if $\phi=0$. Hence, we only need to show that $\|u\|\leq\bar{\varepsilon}:=(2+\lambda)\varepsilon$ for all nonzero $\phi\in\mathcal{C}([-\Delta, 0], \mathbf{B}(\delta))$, where $\lambda>0$ is given in \eqref{eqn-11}. For this purpose, we consider the following two cases. If $L_{g}V_{\rzmk}(\phi)=0$, then $u=0$ from \eqref{eqn-12}, and thus $\|u\|\leq\bar{\varepsilon}$ for all nonzero $\phi\in\mathcal{C}([-\Delta, 0], \mathbf{B}(\delta))$. If $L_{g}V_{\rzmk}(\phi)\neq0$, then we have from \eqref{eqn-14} that for any nonzero $\phi\in\mathcal{C}([-\Delta, 0], \mathbf{B}(\delta))$,
\begin{align*}
L_{f}V_{\rzmk}(\phi)+\gamma_{\rzmk}(V_{\rzmk}(x))<\varepsilon\|L_{g}V_{\rzmk}(\phi)\|+\eta_{\rzmk}(\|V_{\rzmk}(\phi)\|).
\end{align*}
which implies $|\mathfrak{a}_{\rzmk}(\phi)|<\varepsilon\|L_{g}V_{\rzmk}(\phi)\|$. Furthermore, from \eqref{eqn-12} and \eqref{eqn-15}, we have that for any nonzero $\phi\in\mathcal{C}([-\Delta, 0], \mathbf{B}(\delta))$,
\begin{align*}
\|u(\phi)\|&=\left\|\frac{\mathfrak{a}_{\rzmk}(\phi)+\sqrt{\mathfrak{a}^{2}_{\rzmk}(\phi)+\lambda\|L_{g}V_{\rzmk}(\phi)\|^{4}}}{\|L_{g}V_{\rzmk}(\phi)\|^{2}}L_{g}V_{\rzmk}(\phi)\right\| \\
&\leq\frac{2|\mathfrak{a}_{\rzmk}(\phi)|+\sqrt{\lambda}\|L_{g}V_{\rzmk}(\phi)\|^{2}}{\|L_{g}V_{\rzmk}(\phi)\|}  \\
&\leq(2+\sqrt{\lambda})\varepsilon=:\bar{\varepsilon}.
\end{align*}
Therefore, if the R-SCP is satisfied, then the control input is bounded and the bound is a function of $\varepsilon\in\mathbb{R}^{+}$. Since $\varepsilon\in\mathbb{R}^{+}$ can be chosen to be arbitrarily small and $\bar{\varepsilon}\rightarrow0$ as $\varepsilon\rightarrow0$, we conclude that the controller is continuous at the origin. Therefore, the proof is completed.
\hfill$\blacksquare$
\end{proof}

Note that the controller \eqref{eqn-11} is continuous for both the case of nonzero $\phi$ and the case where $\phi(\theta)=0$ for all $\theta\in[-\Delta, 0]$. In addition, if $g(0)=0$, then $L_{g}V_{\rzmk}(0)=0$ and thus $u(0)=0$ from \eqref{eqn-12}, which implies that the case $\phi=0$ can be removed from \eqref{eqn-11}. Finally, the controller \eqref{eqn-11} follows Sontag's formula \citep{Sontag1989universal} and extends the results in \citep{Sontag1989universal, Romdlony2016stabilization, Wu2019control} from the delay-free case to the time-delay case.

\subsection{Razumikhin-type Control Barrier Functions}
\label{subsec-cbrf}

To investigate the safety of the system \eqref{eqn-1}, we propose control barrier-Razumikhin functions (CBRFs) in this subsection. To this end, we define a function $h\in\mathcal{C}(\mathbb{R}^{n}, \mathbb{R})$ to connect with the safe set $\mathbb{C}_{\rzmk}$, where the subscript `$\rzmk$' is to show that it is Razumikhin-type. That is, the set $\mathbb{C}_{\rzmk}$ is defined as
\begin{align}
\label{eqn-16}
\mathbb{C}_{\rzmk}&:=\{x\in\mathbb{R}^{n}: h_{\rzmk}(x)\geq0\}, \\
\label{eqn-17}
\partial\mathbb{C}_{\rzmk}&:=\{x\in\mathbb{R}^{n}: h_{\rzmk}(x)=0\}, \\
\label{eqn-18}
\inte(\mathbb{C}_{\rzmk})&:=\{x\in\mathbb{R}^{n}: h_{\rzmk}(x)>0\}.
\end{align}
Assume that $\inte(\mathbb{C}_{\rzmk})\neq\varnothing$ and $\overbar{\inte(\mathbb{C}_{\rzmk})}=\mathbb{C}_{\rzmk}$, which mean that the set $\mathbb{C}_{\rzmk}$ is nonempty and has no isolated points. In the following, two types of CBRFs are presented.

\begin{definition}
\label{def-5}
Consider the system \eqref{eqn-1} and the set $\mathbb{C}_{\rzmk}\subset\mathbb{R}^{n}$ in \eqref{eqn-16}, a function $B_{\rzmk}\in\mathcal{C}(\inte(\mathbb{C}_{\rzmk}), \mathbb{R})$ is called an \emph{R-type control barrier-Razumikhin function (R-CBRF)} for the set $\mathbb{C}_{\rzmk}$, if
\begin{enumerate}[(i)]
\item there exist $\alpha_{1}, \alpha_{2}\in\mathcal{K}$ such that for all $x\in\inte(\mathbb{C}_{\rzmk})$,
\begin{align}
\label{eqn-19}
\alpha_{1}(h_{\rzmk}(x))\leq\frac{1}{B_{\rzmk}(x)}\leq\alpha_{2}(h_{\rzmk}(x));
\end{align}
\item there exist $\gamma_{\rbrf}, \eta_{\rbrf}\in\overbar{\mathcal{K}}$ such that $\gamma_{\rbrf}-\eta_{\rbrf}\in\mathcal{K}$ and for any nonzero $\phi\in\mathcal{C}([-\Delta, 0], \inte(\mathbb{C}_{\rzmk}))$ with $\phi(0)=x$,
\begin{align}
\label{eqn-20}
&\inf_{u\in\mathbb{U}}\left\{L_{f}B_{\rzmk}(\phi)+L_{g}B_{\rzmk}(\phi)u\right\}<\gamma_{\rbrf}(h_{\rzmk}(x))-\eta_{\rbrf}(\|h_{\rzmk}(\phi)\|).
\end{align}
\end{enumerate}
\end{definition}

Definition \ref{def-5} is based on the function $h_{\rzmk}$ implicitly, and the relation between the functions $B_{\rzmk}$ and $h_{\rzmk}$ is given in \eqref{eqn-19}. The following CBRF is defined from the function $h_{\rzmk}$ explicitly.

\begin{definition}
\label{def-6}
Consider the system \eqref{eqn-1}, a function $h_{\rzmk}\in\mathcal{C}(\mathbb{R}^{n}, \mathbb{R})$ is called a \emph{Z-type control barrier-Razumikhin function (Z-CBRF)} for the set $\mathbb{C}_{\rzmk}\subset\mathbb{R}^{n}$ in \eqref{eqn-16}, if there exist a set $\mathbb{X}_{\rzmk}\subset\mathbb{R}^{n}$ with $\mathbb{C}_{\rzmk}\subseteq\mathbb{X}_{\rzmk}$ and $\gamma_{\zbrf}, \eta_{\zbrf}\in\overbar{\mathcal{K}}$ such that for any nonzero $\phi\in\mathcal{C}([-\Delta, 0], \mathbb{X}_{\rzmk})$ with $\phi(0)=x$,
\begin{align}
\label{eqn-21}
&\sup_{u\in\mathbb{U}}\left\{L_{f}h_{\rzmk}(\phi)+L_{g}h_{\rzmk}(\phi)u\right\}>-\gamma_{\zbrf}(h_{\rzmk}(x))+\eta_{\zbrf}(\|h_{\rzmk}(\phi)\|).
\end{align}
\end{definition}

\begin{remark}
\label{rmk-1}
The proposed CBRFs extend the CBF from the delay-free case \citep{Ames2016control} to the time-delay case. Similar to the CLRF-II, the effects of time delays are transformed into the disturbance-like items in \eqref{eqn-20}-\eqref{eqn-21}. Comparing with the Z-CBRF, the R-CBRF is not based on the function $h_{\rzmk}$ explicitly. Hence, the R-CBRF provides more potential flexibilities in its construction \citep{Lindemann2018control, Panagou2015distributed}. On the other hand, similar to the delay-free case \citep{Ames2016control}, the R-CBRF and Z-CBRF can be transformed to each other under some conditions. Specifically, if $\mathbb{C}_{\rzmk}\equiv\mathbb{X}_{\rzmk}$ and $\dot{h}_{\rzmk}(x)>0$ for all $x\in\partial\mathbb{C}_{\rzmk}$, then $h_{\rzmk}$ is a Z-CBRF if and only if $h^{-1}_{\rzmk}$ is an R-CBRF, which can be verified via the fact that $-\inf\{\mathcal{A}\}=\sup\{-\mathcal{A}\}$ for any compact set $\mathcal{A}\subseteq\mathbb{R}$.
\hfill $\square$
\end{remark}

To show the guarantee of the safety objective via the CBRFs in Definitions \ref{def-5} and \ref{def-6}, we define the following sets:
\begin{align}
\label{eqn-22}
\mathbb{K}_{\rbrf}&:=\left\{u\in\mathbb{U}: L_{f}B_{\rzmk}(\phi)+L_{g}B_{\rzmk}(\phi)u<\gamma_{\zbrf}(h_{\rzmk}(x))\right. \nonumber \\
&\quad \left. -\eta_{\zbrf}(\|V_{\rzmk}(\phi)\|) \text{ or } u(0)=0\right\}, \\
\label{eqn-23}
\mathbb{K}_{\zbrf}&:=\left\{u\in\mathbb{U}: L_{f}h_{\rzmk}(\phi)+L_{g}h_{\rzmk}(\phi)u>-\gamma_{\rbrf}(h_{\rzmk}(x))  \right. \nonumber \\
&\quad \left.+\eta_{\rbrf}(\|h_{\rzmk}(\phi)\|) \text{ or } u(0)=0 \right\}.
\end{align}
The non-emptiness of these two sets follows from similar arguments as those in Section \ref{subsec-clrf}, and thus is omitted here. With these two sets, the following theorem shows that the control inputs from the sets $\mathbb{K}_{\zbrf}$ and $\mathbb{K}_{\rbrf}$ result in the forward invariant property of the safe set $\mathbb{C}_{\rzmk}\subset\mathbb{R}^{n}$.

\begin{theorem}
\label{thm-4}
Consider the system \eqref{eqn-1} with the safe set $\mathbb{C}_{\rzmk}\subset\mathbb{R}^{n}$ in \eqref{eqn-16},
\begin{enumerate}[(i)]
  \item if the system \eqref{eqn-1} admits an R-CBRF $B_{\rzmk}\in\mathcal{C}(\inte(\mathbb{C}_{\rzmk}), \mathbb{R})$, and there exists a Lipschitz continuous functional $u: \mathcal{C}([-\Delta, 0], \inte(\mathbb{C}_{\rzmk}))\rightarrow\mathbb{U}$ such that $u(\phi)\in\mathbb{K}_{\rbrf}$, then the controller $u(\phi)$ ensures the set $\inte(\mathbb{C}_{\rzmk})$ to be forward invariant;

  \item  if the system \eqref{eqn-1} admits a Z-CBRF $h_{\rzmk}\in\mathcal{C}(\mathbb{X}_{\rzmk}, \mathbb{R})$, and there exists a Lipschitz continuous functional $u: \mathcal{C}([-\Delta, 0], \mathbb{X}_{\rzmk})$ $\rightarrow\mathbb{U}$ such that $u(\phi)\in\mathbb{K}_{\zbrf}$, then the controller $u(\phi)$ ensures the set $\mathbb{C}_{\rzmk}$ to be forward invariant.
\end{enumerate}
\end{theorem}

\begin{proof}
(\textbf{First Claim.}) Given the R-CBRF $B_{\rzmk}: \inte(\mathbb{C}_{\rzmk})\rightarrow\mathbb{R}$, we define the function $\Theta_{\rzmk}(x):=1/B_{\rzmk}(x)$, and then have
\begin{align}
\label{eqn-24}
\dot{\Theta}_{\rzmk}(x)&=-(B_{\rzmk}(x))^{-2}\dot{B}_{\rzmk}(x)=-\Theta^{2}_{\rzmk}(x)(L_{f}B_{\rzmk}(\phi)+L_{g}B_{\rzmk}(\phi)u) \nonumber\\
&>-\Theta^{2}_{\rzmk}(x)(\gamma_{\rbrf}(h_{\rzmk}(x))-\eta_{\rbrf}(\|h_{\rzmk}(\phi)\|) \nonumber\\
&>-\Theta^{2}_{\rzmk}(x)[\gamma_{\rbrf}(\alpha^{-1}_{1}(\Theta_{\rzmk}(x)))-\eta_{\rbrf}(\alpha^{-1}_{2}(\|\Theta_{\rzmk}(\phi)\|))].
\end{align}
To emphasize the time argument, let $\Theta_{\rzmk}(t):=\Theta_{\rzmk}(x(t))$. Define $\varphi(\Theta_{\rzmk}(t)):=\Theta^{2}_{\rzmk}(t)\gamma_{\rbrf}(\alpha^{-1}_{1}(\Theta_{\rzmk}(t)))$ and $\psi(\Theta_{\rzmk}(t+\theta)):=\Theta^{2}_{\rzmk}(t)\eta_{\rbrf}(\alpha^{-1}_{2}(\Theta_{\rzmk}(t+\theta)))$. From \eqref{eqn-24}, we consider the following inequality and equation:
\begin{align}
\label{eqn-25}
\dot{\Theta}_{\rzmk}(t)&>-\varphi(\Theta_{\rzmk}(t))+\sup_{\theta\in[-\Delta, 0]}\psi(\Theta_{\rzmk}(t+\theta)), \\
\label{eqn-26}
\dot{\Pi}_{\rzmk}(t)&=-\varphi(\Pi_{\rzmk}(t)),
\end{align}
where $t\geq0$ and $\Theta_{\rzmk}(\theta)=\Gamma_{\rzmk}(\theta)$ for all $\theta\in[-\Delta, 0]$. From Lemma 4.4 in \citep{Khalil2002nonlinear}, there exists $\beta_{1}\in\mathcal{KL}$ such that the solution of \eqref{eqn-26} satisfies $\Pi_{\rzmk}(t)=\beta_{1}(\Pi_{\rzmk}(0), t)$  for all $t\geq0$. Using Lemma \ref{lem-A1} in \ref{sec-appendix} and from \eqref{eqn-25}-\eqref{eqn-26}, we have
\begin{align*}
\Theta_{\rzmk}(t)&\geq\Pi_{\rzmk}(t)=\beta_{1}(\Pi_{\rzmk}(0), t), \quad \forall t\geq0,
\end{align*}
combining which with the definition of $\Theta_{\rzmk}(t)$ yields
\begin{align}
\label{eqn-27}
\frac{1}{B_{\rzmk}(x(t))}&\geq\beta_{1}\left(\frac{1}{B_{\rzmk}(x(0))}, t\right), \quad  \forall t\geq0.
\end{align}
From \eqref{eqn-19} and \eqref{eqn-27}, we obtain that for all $t\geq0$,
\begin{align*}
h_{\rzmk}(x(t))\geq\alpha^{-1}_{2}(\beta_{1}(\alpha_{1}(h_{\rzmk}(x(0))), t)).
\end{align*}
That is, for all $t>0$, $h_{\rzmk}(x(t))>0$, and thus $\inte(\mathbb{C}_{\rzmk})$ is forward invariant.

(\textbf{Second Claim.}) From the Z-CBRF and $u\in\mathbb{K}_{\zbrf}$, there exists a set $\mathbb{X}_{\rzmk}\subset\mathbb{R}^{n}$ such that $\mathbb{C}_{\rzmk}\subset\mathbb{X}_{\rzmk}$ and for all $\phi\in\mathcal{C}([-\Delta, 0], \mathbb{X}_{\rzmk})$ with $\phi(0)=x$,
\begin{align*}
&L_{f}h_{\rzmk}(\phi)+L_{g}h_{\rzmk}(\phi)u>-\gamma_{\zbrf}(h_{\rzmk}(x))+\eta_{\zbrf}(\|h_{\rzmk}(\phi)\|),
\end{align*}
which implies that $\dot{h}_{\rzmk}(x)\geq-\gamma_{\zbrf}(h_{\rzmk}(x))+\eta_{\zbrf}(\|h_{\rzmk}(\phi)\|)$. Let $h_{\rzmk}(t):=h_{\rzmk}(x(t))$ to emphasize the time argument. Consider the following inequality and equation:
\begin{align}
\label{eqn-28}
\dot{h}_{\rzmk}(t)&>-\gamma_{\zbrf}(h_{\rzmk}(t))+\sup_{\theta\in[-\Delta, 0]}\eta_{\zbrf}(h_{\rzmk}(t+\theta)), \\
\label{eqn-29}
\dot{\Phi}_{\rzmk}(t)&=-\gamma_{\zbrf}(\Phi_{\rzmk}(t)),
\end{align}
where $t>0$ and $h_{\rzmk}(\theta)=\Phi_{\rzmk}(\theta)$ for all $\theta\in[-\Delta, 0]$. From Lemma 4.4 of \citep{Khalil2002nonlinear}, there exists $\beta_{2}\in\mathcal{KL}$ such that the solution to \eqref{eqn-29} satisfies $\Phi_{\rzmk}(t)=\beta_{2}(\Phi_{\rzmk}(0), t)$ for all $t\geq0$. From \eqref{eqn-28}-\eqref{eqn-29} and using Lemma \ref{lem-A1} in \ref{sec-appendix}, we have
\begin{align*}
h_{\rzmk}(x)\geq\Phi_{\rzmk}(t)=\beta_{2}(h_{\rzmk}(x(0)), t), \quad \forall t\geq0, \forall x(0)\in\mathbb{C}_{\rzmk}.
\end{align*}
Hence, we conclude that $\mathbb{C}_{\rzmk}$ is forward invariant.
\hfill$\blacksquare$
\end{proof}

\section{Krasovskii-type Control Functionals}
\label{sec-Krasovskiitype}

Different from the Razumikhin approach where the construction of control functions does not depend on the time-delay trajectory, another type of control functions for time-delay systems is based on the Krasovskii approach, which is related to the time-delay trajectory \citep{Jankovic2001control, Pepe2014stabilization, Jankovic2018control}. Using the Krasovskii approach, we propose control Lyapunov-Krasovskii functionals (CLKFs) and control barrier-Krasovskii functionals (CBKFs) for time-delay systems in this section.

\subsection{Control Lyapunov-Krasovskii Functional}
\label{subsec-clkf}

Before defining CLKFs, we first recall the class of smoothly separable functionals from \citep{Pepe2014stabilization}.

\begin{definition}
\label{def-7}
A functional $V: \mathcal{C}([-\Delta, 0], \mathbb{R}^{n})\rightarrow\mathbb{R}$ is \emph{smoothly separable}, if there exist a function $V_{1}\in\mathcal{C}(\mathbb{R}^{n}, \mathbb{R}^{+})$, a locally Lipschitz functional $V_{2}: \mathcal{C}([-\Delta, 0], \mathbb{R}^{n})\rightarrow\mathbb{R}^{+}$ and $\alpha_{1}, \alpha_{2}\in\mathcal{K}_{\infty}$ such that, for all $\phi\in\mathcal{C}([-\Delta, 0], \mathbb{R}^{n})$,
\begin{align}
\label{eqn-30}
V(\phi)&:=V_{1}(\phi(0))+V_{2}(\phi), \\
\label{eqn-31}
\alpha_{1}(|\phi(0)|)&\leq V_{1}(\phi(0))\leq\alpha_{2}(|\phi(0)|).
\end{align}
\end{definition}

From Definition \ref{def-7}, any smoothly separable functional is locally Lipschitz. In the following, we recall the definition of invariant differentiability from \citep{Kim1999functional}.

\begin{definition}
\label{def-8}
A smoothly separable functional $V:\mathcal{C}([-\Delta,$ $0], \mathbb{R}^{n})\rightarrow\mathbb{R}^{+}$ is \emph{invariantly differentiable (i-differentiable)}, if $V(\phi)=V_{1}(\phi(0))+V_{2}(\phi)$ and for all $\phi\in\mathcal{C}([-\Delta, 0], \mathbb{R}^{n})$,
\begin{enumerate}[(i)]
  \item both $\dot{V}_{1}(\phi(0))$ and $D^{+}V_{2}(\phi)$ exist;
  \item $D^{+}V_{2}(\phi)$ is invariant with respect to $\phi\in\mathcal{C}([-\Delta, 0], \mathbb{R}^{n})$, that is, $D^{+}V_{2}(x_{0})$ is the same for all $x_{t}\in\mathcal{C}([-\Delta, 0], \mathbb{R}^{n})$;
  \item  for all $x, z\in\mathbb{R}^{n}, x_{t}\in\mathcal{C}([-\Delta, 0], \mathbb{R}^{n})$ and $l\geq0$, $V(x+z, x_{t+l})-V(x, x_{t}):=\dot{V}_{1}(x)z+D^{+}V_{2}(x_{t})l+o(\sqrt{|z|^{2}+l^{2}})$,
where $\lim_{s\rightarrow0^{+}}o(s)/s=0$.
\end{enumerate}
In addition, if $D^{+}V_{2}(\phi)$ is continuous, then $V$ is said to be \emph{continuously i-differentiable}.
\end{definition}

From Definition \ref{def-8}, we have that $D^{+}V(\phi):=L_{f}V_{1}(\phi)+L_{g}V_{1}(\phi)u+D^{+}V_{2}(\phi)$ for all $\phi\in\mathcal{C}([-\Delta, 0], \mathbb{R}^{n})$.
Different from \citep{Kim1999functional} where the invariant differentiability is for the functionals on $\mathbb{R}^{n}\times\mathcal{C}([-\Delta, 0), \mathbb{R}^{n})$, Definition \ref{def-8} is for smoothly separable functionals and $V_{2}$ is defined on $\mathcal{C}([-\Delta, 0], \mathbb{R}^{n})$ to ensure the well-posedness of the derivative; see, e.g., Example 2.2.1 in \citep{Kim1999functional}. Definitions \ref{def-7}-\ref{def-8} include many Lyapunov-Krasovskii functionals in existing works \citep{Pepe2016stabilization, Orosz2019safety, Zhou2016razumikhin, Karafyllis2016stabilization, Liu2011input, Prajna2005methods}. With smoothly separable and i-differentiable functionals, the CLKF is proposed as follows.

\begin{definition}
\label{def-9}
For the system \eqref{eqn-1}, a smoothly separable and continuously i-differentiable functional $V_{\ksvk}: \mathcal{C}([-\Delta, 0], \mathbb{R}^{n})\rightarrow\mathbb{R}^{+}$ is called a \emph{control Lyapunov-Krasovskii functional (CLKF)}, if for all $\phi\in\mathcal{C}([-\Delta, 0], \mathbb{R}^{n})$,
\begin{enumerate}[(i)]
  \item there exist $\alpha_{1}, \alpha_{2}\in\mathcal{K}_{\infty}$ such that $\alpha_{1\ksvk}(|\phi(0)|)\leq V_{\ksvk}(\phi)\leq\alpha_{2\ksvk}(\|\phi\|)$;
  \item there exists $\gamma_{\lyaK}\in\mathcal{K}$ such that for any nonzero $\phi\in\mathcal{C}([-\Delta, 0], \mathbb{R}^{n})$,
  \begin{align}
  \label{eqn-32}
  &\inf_{u\in\mathbb{U}}\left\{L_{f}V_{\ksvk1}(\phi)+D^{+}V_{\ksvk2}(\phi)+L_{g}V_{\ksvk1}(\phi)u\right\}<-\gamma_{\lyaK}(V_{\ksvk}(\phi)).
  \end{align}
\end{enumerate}
\end{definition}

Different from the CLRF in Section \ref{subsec-clrf}, the CLKF is defined on $\mathcal{C}([-\Delta, 0], \mathbb{R}^{n})$ such that the time-delay trajectory is involved directly. For instance, if $V_{\ksvk}(\phi)=|\phi(0)|^2+\int^{0}_{-\Delta}|\phi(s)|^2ds$, then $V_{\ksvk1}(\phi(0))=|\phi(0)|^2$ and  $V_{\ksvk2}(\phi)=\int^{0}_{-\Delta}|\phi(s)|^2ds$. In addition, $L_{f}V_{\ksvk1}(\phi)=2\phi^{\top}(0)f(\phi)$, $L_{g}V_{\ksvk1}(\phi)u=2\phi^{\top}(0)g(\phi)u$ and $D^{+}V_{\ksvk2}(\phi)=|\phi(0)|^2-|\phi(-\Delta)|^2$. From \eqref{eqn-31}, we define the set $\mathbb{K}_{\lyaK}:=\{u\in\mathbb{U}: L_{f}V_{\ksvk1}(\phi)+D^{+}V_{\ksvk2}(\phi)+L_{g}V_{\ksvk1}(\phi)u<-\gamma_{\lyaK}(V_{\ksvk}(\phi))\}$. Similar to the mechanism in \eqref{eqn-9}, the set $\mathbb{K}_{\lyaK}$ is non-empty. Therefore, any locally Lipschitz controller in the set $\mathbb{K}_{\lyaK}$ leads to the stabilization of the system \eqref{eqn-1}. Similar to the R-SCP of Definition \ref{def-4}, we propose the Krasovskii-type SCP as follows.

\begin{definition}
\label{def-10}
Consider the system \eqref{eqn-1} with a CLKF $V_{\ksvk}$, the system \eqref{eqn-1} is said to satisfy the \emph{Krasovskii-type small control property (K-SCP)}, if for arbitrary $\varepsilon>0$, there exists $\delta>0$ such that, for any nonzero $\phi\in\mathcal{C}([-\Delta, 0], \mathbf{B}(\delta))$, there exists $u\in\mathbf{B}(\varepsilon)$ such that $L_{f}V_{\ksvk1}(\phi)+D^{+}V_{\ksvk2}(\phi)+L_{g}V_{\ksvk1}(\phi)u<-\gamma_{\lyaK}(V_{\ksvk}(\phi))$.
\end{definition}

Definition \ref{def-10} extends the SCP to the time-delay case in terms of the Krasovskii approach. Definition \ref{def-10} is similar to item (iv) of Hypothesis 8 in \citep{Pepe2016stabilization} for neutral functional differential equations, while different from the one in \citep{Karafyllis2016stabilization}, where $\phi\in\mathcal{C}([-\Delta, 0], \mathbf{B}(\delta))$ with $\phi(0)\neq0$ is required. With the proposed CLKF and K-SCP, the controller is derived in the next theorem to ensure the stabilization of the system \eqref{eqn-1}.

\begin{theorem}
\label{thm-5}
If the system \eqref{eqn-1} admits a CLKF and satisfies the K-SCP, then the following controller
\begin{align}
\label{eqn-33}
u(\phi):=\left\{\begin{aligned}
&\kappa(\lambda, \mathfrak{a}_{\ksvk}(\phi), (L_{g}V_{\ksvk1}(\phi))^{\top}), &\text{ if }& \phi\neq 0, \\
&0, &\text{ if }&  \phi=0,
\end{aligned}\right.
\end{align}
where $\kappa$ is defined in \eqref{eqn-12}, $\lambda>0$ and $\mathfrak{a}_{\ksvk}(\phi):=L_{f}V_{\ksvk1}(\phi)+D^{+}V_{\ksvk2}(\phi)+\gamma_{\lyaK}(V_{\ksvk}(\phi))$, is continuous and ensures the system \eqref{eqn-1} to be GAS.
\end{theorem}

\begin{proof}
We first prove the stabilization of the system \eqref{eqn-1} under the controller \eqref{eqn-33}. From \eqref{eqn-33}, we have that if $L_{g}V_{\ksvk1}(\phi)=0$, then $u=0$, and from \eqref{eqn-32} in Definition \ref{def-9}, we have
\begin{align*}
&L_{f}V_{\ksvk1}(\phi)+D^{+}V_{\ksvk2}(\phi)+L_{g}V_{\ksvk1}(\phi)u+\gamma_{\lyaK}(V_{\ksvk}(\phi)) \\
&=L_{f}V_{\ksvk1}(\phi)+D^{+}V_{\ksvk2}(\phi)+\gamma_{\lyaK}(V_{\ksvk}(\phi))<0.
\end{align*}
If $L_{g}V_{\ksvk1}(\phi)\neq0$, we define $\mathfrak{a}_{\ksvk}(\phi):=L_{f}V_{\ksvk1}(\phi)+D^{+}V_{\ksvk2}(\phi)+\gamma_{\lyaK}(V_{\ksvk}(\phi))$ and $\mathfrak{b}_{\ksvk}(\phi):=L_{g}V_{\ksvk1}(\phi)$. From the control law \eqref{eqn-33},
\begin{align*}
&L_{f}V_{\ksvk1}(\phi)+D^{+}V_{\ksvk2}(\phi)+L_{g}V_{\ksvk1}(\phi)u+\gamma_{\lyaK}(V_{\ksvk}(\phi)) \\
&=\mathfrak{a}_{\ksvk}(\phi)-\mathfrak{b}_{\ksvk}(\phi)\frac{\mathfrak{a}_{\ksvk}(\phi)\mathfrak{b}^{\top}_{\ksvk}(\phi)}{\|\mathfrak{b}_{\ksvk}(\phi)\|^{2}} -\mathfrak{b}_{\ksvk}(\phi)\frac{\sqrt{\mathfrak{a}^{2}_{\ksvk}(\phi)+\lambda\|\mathfrak{b}_{\ksvk}(\phi)\|^{4}}}{\|\mathfrak{b}_{\ksvk}(\phi)\|^{2}}\mathfrak{b}^{\top}_{\ksvk}(\phi) \\
&=-\sqrt{\mathfrak{a}^{2}_{\ksvk}(\phi)+\lambda\|\mathfrak{b}_{\ksvk}(\phi)\|^{4}}\leq0.
\end{align*}
Hence, under the controller \eqref{eqn-33}, $L_{f}V_{\ksvk1}(\phi)+D^{+}V_{\ksvk2}(\phi)+L_{g}V_{\ksvk1}(\phi)u<-\gamma_{\lyaK}(V_{\ksvk}(\phi))$. From \eqref{eqn-30}-\eqref{eqn-31} and Theorem 2.1 in \citep{Hale1993introduction}, we conclude that the closed-loop system is GAS.

Next, we show the continuity of the controller. From the properties of the functional $V_{\ksvk}$ and the functions $f, g$, we deduce the continuity of the controller in any region away from the origin. Hence, we only need to show the continuity of the controller at the origin. From the K-SCP, for arbitrary $\varepsilon\in\mathbb{R}^{+}$, there exists $\delta_{1}\in\mathbb{R}^{+}$ such that, for any nonzero $\phi\in\mathcal{C}([-\Delta, 0], \mathbf{B}(\delta_{1}))$, there exists $u\in\mathbf{B}(\varepsilon)$ such that $\mathfrak{a}_{\ksvk}(\phi)+L_{g}V_{\ksvk1}(\phi)u<0$. Since $V_{1}\in\mathcal{C}(\mathbb{R}^{n}, \mathbb{R}^{+})$ and $g$ in \eqref{eqn-1} is locally Lipschitz, there exists $\delta_{2}\in\mathbb{R}^{+}$ with $\delta_{2}\neq\delta_{1}$ such that for all nonzero $\phi\in\mathcal{C}([-\Delta, 0], \mathbf{B}(\delta_{2}))$, $\|L_{g}V_{\ksvk1}(\phi)\|\leq\varepsilon$. Let $\delta:=\min\{\delta_{1}, \delta_{2}\}$. For any nonzero $\phi\in\mathcal{C}([-\Delta, 0], \mathbf{B}(\delta))$, we consider the following two cases. The first case is $L_{g}V_{\ksvk1}(\phi)=0$. In this case, $u(\phi)=0$ from \eqref{eqn-33}. In addition, $u(0)=0$ from \eqref{eqn-33}, and thus $\|u(\phi)-u(0)\|=0<\varepsilon$. Hence, if the K-SCP is satisfied, then the control input is bounded by $\varepsilon\in\mathbb{R}^{+}$. Since $\varepsilon\in\mathbb{R}^{+}$ can be chosen arbitrarily small, we conclude that the controller is continuous in this case.
The second case is $L_{g}V_{\ksvk1}(\phi)\neq0$. In this case, $|\mathfrak{a}_{\ksvk}(\phi)|\leq\varepsilon\|L_{g}V_{\ksvk1}(\phi)\|$ for any nonzero $\phi\in\mathcal{C}([-\Delta, 0], \mathbf{B}(\delta))$. From \eqref{eqn-33}, we have that for any nonzero $\phi\in\mathcal{C}([-\Delta, 0], \mathbf{B}(\delta))$,
\begin{align*}
\|u(\phi)\|&\leq\left\|\frac{\mathfrak{a}_{\ksvk}(\phi)+\sqrt{\mathfrak{a}^{2}_{\ksvk}(\phi)+\lambda\|L_{g}V_{\ksvk1}(\phi)\|^{4}}}{\|L_{g}V_{\ksvk1}(\phi)\|^{2}}L_{g}V_{\ksvk1}(\phi)\right\| \\
&\leq\left|\frac{\mathfrak{a}_{\ksvk}(\phi)+\sqrt{\mathfrak{a}^{2}_{\ksvk}(\phi)+\lambda\|L_{g}V_{\ksvk1}(\phi)\|^{4}}}{\|L_{g}V_{\ksvk1}(\phi)\|}\right| \\
&\leq\left|\frac{\mathfrak{a}_{\ksvk}(\phi)+|\mathfrak{a}_{\ksvk}(\phi)|+\sqrt{\lambda}\|L_{g}V_{\ksvk1}(\phi)\|^{2}}{\|L_{g}V_{\ksvk1}(\phi)\|}\right| \\
&\leq\sqrt{\lambda}\varepsilon+\left|\frac{\mathfrak{a}_{\ksvk}(\phi)+|\mathfrak{a}_{\ksvk}(\phi)|}{\|L_{g}V_{\ksvk1}(\phi)\|}\right|,
\end{align*}
which implies that $\|u(\phi)\|\leq\sqrt{\lambda}\varepsilon$ if $\mathfrak{a}_{\ksvk}(\phi)\leq0$ and $\|u(\phi)\|\leq(2+\sqrt{\lambda})\varepsilon$ if $\mathfrak{a}_{\ksvk}(\phi)>0$.
Since $\lim_{\varepsilon\rightarrow0}\sqrt{\lambda}\varepsilon=\lim_{\varepsilon\rightarrow0}(2+\sqrt{\lambda})\varepsilon=0$ and $\varepsilon\in\mathbb{R}^{+}$ can be arbitrarily small, the controller is continuous at the origin in the second case. Therefore, from these two cases, we conclude that the controller is continuous at the origin, which completes the proof.
\hfill$\blacksquare$
\end{proof}

\subsection{Krasovskii-type Control Barrier Functionals}
\label{subsec-cbkf}

To address the safety of the system \eqref{eqn-1} in terms of the Krasovskii approach, we first redefine the safe set, and then propose control barrier-Krasovskii functionals. For the safe set $\mathbb{C}_{\ksvk}$, a continuously differentiable functional $h_{\ksvk}: \mathcal{C}([-\Delta, 0], \mathbb{R}^{n})\rightarrow\mathbb{R}$ is used here to describe $\mathbb{C}_{\ksvk}$, and further to determine the forward invariant property. Assume that $\mathbb{C}_{\ksvk}\subset\mathbb{R}^{n}$ is defined as
\begin{align}
\label{eqn-34}
\mathbb{C}_{\ksvk}&:=\{\phi\in\mathcal{C}([-\Delta, 0], \mathbb{R}^{n}): h_{\ksvk}(\phi)\geq0\}, \\
\label{eqn-35}
\partial\mathbb{C}_{\ksvk}&:=\{\phi\in\mathcal{C}([-\Delta, 0], \mathbb{R}^{n}): h_{\ksvk}(\phi)=0\}, \\
\label{eqn-36}
\inte(\mathbb{C}_{\ksvk})&:=\{\phi\in\mathcal{C}([-\Delta, 0], \mathbb{R}^{n}): h_{\ksvk}(\phi)>0\}.
\end{align}
Let $\inte(\mathbb{C}_{\ksvk})\neq\varnothing$ and $\overbar{\inte(\mathbb{C}_{\ksvk})}=\mathbb{C}_{\ksvk}$. Comparing with \eqref{eqn-16}-\eqref{eqn-18} based on the current state, the safe set in \eqref{eqn-34}-\eqref{eqn-36} is defined via the time-delay trajectory and similar to the one in \citep{Orosz2019safety}.

\begin{definition}
\label{def-11}
Consider the system \eqref{eqn-1} and the set $\mathbb{C}_{\ksvk}$ in \eqref{eqn-34}, a smoothly separable and continuously i-differentiable functional $B_{\ksvk}: \mathcal{C}([-\Delta, 0],$ $\mathbb{R}^{n})\rightarrow\mathbb{R}$ is called an \emph{R-type control barrier-Krasovskii functional (R-CBKF)} for $\mathbb{C}_{\ksvk}$, if
\begin{enumerate}[(i)]
\item there exist $\alpha_{1}, \alpha_{2}\in\mathcal{K}$ such that, for all $\phi\in\inte(\mathbb{C}_{\ksvk})$,
\begin{align}
\label{eqn-37}
\alpha_{1}(h_{\ksvk}(\phi))\leq\frac{1}{B_{\ksvk}(\phi)}\leq\alpha_{2}(h_{\ksvk}(\phi));
\end{align}

\item there exists $\gamma_{\rbkf}\in\mathcal{K}$ such that $\inf_{u\in\mathbb{U}}\{L_{f}B_{\ksvk1}(\phi)+D^{+}B_{\ksvk2}(\phi)+L_{g}B_{\ksvk1}(\phi)u\}<\gamma_{\rbkf}(h_{\ksvk}(\phi))$ for any nonzero $\phi\in\mathcal{C}([-\Delta, 0], \mathbb{R}^{n})$.
\end{enumerate}
\end{definition}

As a counterpart of Definition \ref{def-8}, the following definition is presented based on the functional $h_{\ksvk}$ in \eqref{eqn-34}-\eqref{eqn-36} directly.

\begin{definition}
\label{def-12}
Consider the system \eqref{eqn-1}, a smoothly separable and continuously i-differentiable functional $h_{\ksvk}: \mathcal{C}([-\Delta, 0], \mathbb{R}^{n})$ $\rightarrow\mathbb{R}$ is called a \emph{Z-type control barrier-Krasovskii functional (Z-CBKF)} for the set $\mathbb{C}_{\ksvk}$ in \eqref{eqn-34}, if there exists $\mathbb{X}_{\ksvk}\in\mathcal{C}([-\Delta, 0], \mathbb{R}^{n})$ such that $\mathbb{C}_{\ksvk}\subseteq\mathbb{X}_{\ksvk}$, and for any nonzero $\phi\in\mathbb{X}_{\ksvk}$,
\begin{align}
\label{eqn-38}
&\sup_{u\in\mathbb{U}}\left\{L_{f}h_{\ksvk1}(\phi)+D^{+}h_{\ksvk2}(\phi)+L_{g}h_{\ksvk1}(\phi)u\right\}>-\gamma_{\zbkf}(h_{\ksvk}(\phi)).
\end{align}
\end{definition}

\begin{remark}
\label{rmk-2}
Following the Krasovskii approach, both the R-CBKF and Z-CBKF extend the CBF to the time-delay case. Different from the CBRFs in Section \ref{subsec-cbrf} where the effects of time delays are transformed into the disturbance-like items, the CBKFs are defined to involve the time-delay trajectory directly. In addition, both the R-CBKF and Z-CBKF are required to be smoothly separable, and this requirement results in differences in defining the safe set since $h_{\ksvk1}(\phi(0))+h_{\ksvk2}(\phi)\geq0$ in \eqref{eqn-34} whereas $h_{\rzmk}(\phi(0))\geq0$ in \eqref{eqn-16}. Furthermore, similar to the R-CBRF, the R-CBKF offers more potential flexibilities in its construction comparing with the Z-RBKF. The mutual transformation between the R-CBKF and Z-CBKF is available under some conditions, which are omitted here.
\hfill $\square$
\end{remark}

From Definitions \ref{def-11}-\ref{def-12}, we define the following sets
\begin{align}
\label{eqn-39}
\mathbb{K}_{\rbkf}&:=\left\{u\in\mathbb{U}: L_{f}B_{\ksvk1}(\phi)+D^{+}B_{\ksvk2}(\phi)\right. \nonumber \\
& \left.\qquad +L_{g}B_{\ksvk1}(\phi)u<\gamma_{\rbkf}(h_{\ksvk}(\phi))\right\}, \\
\label{eqn-40}
\mathbb{K}_{\zbkf}&:=\left\{u\in\mathbb{U}: L_{f}h_{\ksvk1}(\phi)+D^{+}h_{\ksvk2}(\phi) \right. \nonumber  \\
& \left.\qquad +L_{g}h_{\ksvk1}(\phi)u>-\gamma_{\zbkf}(h_{\ksvk}(\phi))\right\}.
\end{align}
In the next theorem, we show that the control inputs from the sets $\mathbb{K}_{\rbkf}$ and $\mathbb{K}_{\zbkf}$ lead to the forward invariant property of the set $\mathbb{C}_{\ksvk}$, which thus implies the satisfaction of the safety objective via the CBKFs.

\begin{theorem}
\label{thm-6}
Consider the system \eqref{eqn-1} and the set $\mathbb{C}_{\ksvk}$ in \eqref{eqn-34},
\begin{enumerate}[(i)]
  \item if the system \eqref{eqn-1} admits an R-CBKF $B_{\ksvk}: \inte(\mathbb{C}_{\ksvk})\rightarrow\mathbb{R}$, and there exists a Lipschitz continuous functional $u: \inte(\mathbb{C}_{\ksvk})\rightarrow\mathbb{U}$ such that $u(\phi)\in\mathbb{K}_{\rbkf}$, then the controller $u(\phi)$ ensures the set $\inte(\mathbb{C}_{\ksvk})$ to be forward invariant;

  \item  if the system \eqref{eqn-1} admits a Z-CBKF $h_{\ksvk}: \mathbb{X}_{\ksvk}\rightarrow\mathbb{R}$, and there exists a Lipschitz continuous functional $u: \mathbb{X}_{\ksvk}\rightarrow\mathbb{U}$ such that $u(\phi)\in\mathbb{K}_{\zbkf}$, then the controller $u(\phi)$ ensures the set $\mathbb{C}_{\ksvk}$ to be forward invariant.
\end{enumerate}
\end{theorem}

\begin{proof}
(\textbf{First Claim.}) For the R-CBKF $B_{\ksvk}: \inte(\mathbb{C}_{\ksvk})\rightarrow\mathbb{R}$, we define $\Theta_{\ksvk}(\phi):=1/B_{\ksvk}(\phi)$, and have
\begin{align}
\label{eqn-41}
D^{+}\Theta_{\ksvk}(\phi)&=-\frac{D^{+}B_{\ksvk}(\phi)}{B^{2}_{\ksvk}(\phi)} \nonumber \\
&=-\Theta^{2}_{\ksvk}(\phi)(L_{f}B_{\ksvk1}(\phi)+D^{+}B_{\ksvk2}(\phi)+L_{g}B_{\ksvk1}(\phi)u) \nonumber\\
&>-\Theta^{2}_{\ksvk}(\phi)\gamma_{\rbkf}(h_{\ksvk}(\phi)) \nonumber \\
&\geq-\Theta^{2}_{\ksvk}(\phi)\gamma_{\rbkf}(\alpha^{-1}_{1}(\Theta_{\ksvk}(\phi))),
\end{align}
where the first ``$\geq$'' holds from Definition \ref{def-11}; and the second ``$\geq$'' holds from \eqref{eqn-37}. Let $\beta(\Theta_{\ksvk}(\phi)):=\Theta^{2}_{\ksvk}(\phi)\gamma_{\rbkf}(\alpha^{-1}_{1}(\Theta_{\ksvk}(\phi)))$, which is of class $\mathcal{K}$. Since $B_{\ksvk}$ is continuously i-differentiable, $\Theta_{\ksvk}$ is continuously i-differentiable from Theorem 6.2 in \cite{Kim2015smooth}. Define $\mathbf{X}(t)=\Theta_{\ksvk}(x_{t})$ for all $t\in\mathbb{R}^{+}$. From Lemma 5 in \citep{Pepe2016stabilization}, $D^{+}\mathbf{X}(t)=D^{+}\Theta_{\ksvk}(x_{t})$, and thus $D^{+}\mathbf{X}(t)>-\beta(\mathbf{X}(t))$ for all $t\in\mathbb{R}^{+}$. From Lemma 4.4 in \citep{Khalil2002nonlinear}, there exists $\sigma_{1}\in\mathcal{KL}$ such that
\begin{align}
\label{eqn-42}
\Theta_{\ksvk}(x_{t})\geq\sigma_{1}(\Theta_{\ksvk}(x_{0}), t), \quad  \forall t\geq0.
\end{align}
Combining \eqref{eqn-42} with the definition of $\Theta_{\ksvk}(\phi)$ implies that
\begin{align*}
\frac{1}{B_{\ksvk}(x_{t})}&\geq\sigma_{1}\left(\frac{1}{B_{\ksvk}(x_{0})}, t\right), \quad  \forall t\geq0,
\end{align*}
which further implies from item (i) in Definition \ref{def-11} that
\begin{equation*}
h_{\ksvk}(x_{t})\geq\alpha^{-1}_{2}(\sigma_{1}(\alpha_{1}(h_{\ksvk}(x_{0})), t)), \quad  \forall t\geq0.
\end{equation*}
Since $h_{\ksvk}(x_{0})>0$, $h_{\ksvk}(x_{t})>0$ for all $t>0$, and thus the set $\inte(\mathbb{C}_{\ksvk})$ is forward invariant.

(\textbf{Second Claim.}) From the Z-CBKF and the controller $u\in\mathbb{K}_{\zbrf}$, there exists a set $\mathbb{X}_{\ksvk}\in\mathcal{C}([-\Delta, 0], \mathbb{R}^{n})$ with $\mathbb{C}_{\ksvk}\subset\mathbb{X}_{\ksvk}$ such that for all $\phi\in\mathbb{X}_{\ksvk}$,
\begin{equation*}
L_{f}h_{\ksvk1}(\phi)+D^{+}h_{\ksvk2}(\phi)+L_{g}h_{\ksvk1}(\phi)u>-\gamma_{\zbkf}(h_{\ksvk}(\phi)),
\end{equation*}
which implies
\begin{align}
\label{eqn-43}
D^{+}h_{\ksvk}(\phi)>-\gamma_{\zbkf}(h_{\ksvk}(\phi)), \quad \forall \phi\in\mathbb{X}_{\ksvk}.
\end{align}
Let $\mathbf{Y}(t):=h_{\ksvk}(x_{t})$ for all $t\in\mathbb{R}^{+}$. Similarly, we have from Lemma 5 in \citep{Pepe2016stabilization} that $D^{+}\mathbf{Y}(t)=D^{+}h_{\ksvk}(x_{t})$ and further from \eqref{eqn-43} that $D^{+}\mathbf{Y}(t)\geq-\gamma_{\zbkf}(\mathbf{Y}(t))$ for all $t\in\mathbb{R}^{+}$. From Lemma 4.4 in \citep{Khalil2002nonlinear}, there exists $\sigma_{2}\in\mathcal{KL}$ such that
\begin{align*}
\mathbf{Y}(t)=h_{\ksvk}(x_{t})\geq\sigma_{2}(h_{\ksvk}(x_{0}), t), \quad \forall t\geq0,
\end{align*}
which indicates that $\mathbb{C}_{\ksvk}$ is forward invariant.
\hfill$\blacksquare$
\end{proof}

In Sections \ref{sec-Razumikhintype}-\ref{sec-Krasovskiitype}, the stabilization and safety control problems are studied individually via the proposed control functions/functionals. In particular, the closed-form stabilizing controller is established for time-delay systems, whereas the safety objective is verified via the existence of the CLRFs and CLKFs. To study both stabilization and safety objectives, how to combine the proposed control functions/functionals and how to design the feedback controller are investigated in the next section.

\section{Sliding Mode Control based Combination}
\label{sec-combinedfunction}

To guarantee the safety and stabilization of time-delay systems simultaneously, the proposed control functions/functionals can be combined via quadratic programming \citep{Jankovic2018robust, Ames2016control} for the delay-free case. However, it is not easy to obtain analytical solutions for time-delay optimal control problems \citep{Wu2019new}. In this section, we combine the proposed control functions/functionals via sliding surface functionals, which allow for the transformation of the proposed control functions/functionals via different combination techniques.

\subsection{Properties of Sliding Surface Functional}

In the applied combination approach, the essence is how to develop the sliding surface functional. To this end, we assume that $\mathcal{V}$ and $\mathcal{B}$ are respectively the CLF and CBF for the system \eqref{eqn-1}. That is, $\mathcal{V}$ and $\mathcal{B}$ can be either Razumikhin-type or Krasovskii-type. Based on the CLF $\mathcal{V}$ and CBF $\mathcal{B}$, the sliding surface functional is defined as
\begin{align}
\label{eqn-44}
\mathbf{U}(\phi)&:=\Upsilon(\mathcal{V}(\phi), \mathcal{B}(\phi)),
\end{align}
where $\mathbf{U}: \phi\mapsto\mathbb{R}$ and $\Upsilon: \mathbb{R}\times\mathbb{R}\rightarrow\mathbb{R}$ are continuously differentiable. Define
\begin{equation*}
D^{+}\mathbf{U}(\phi):=\mathbf{F}(\phi)+\mathbf{G}(\phi)u+\mathbf{L}(\phi)
\end{equation*}
with $\mathbf{F}(\phi)=\mathbf{H}(\phi)f(\phi)$ and $\mathbf{G}(\phi)=\mathbf{H}(\phi)g(\phi)$. To be specific, if $\mathbf{U}(x):=\Upsilon(\mathcal{V}(x), \mathcal{B}(x))$, then
\begin{equation*}
\mathbf{H}(x)=\frac{\partial\Upsilon}{\partial\mathcal{V}}\frac{\partial\mathcal{V}}{\partial x}+\frac{\partial\Upsilon}{\partial\mathcal{B}}\frac{\partial\mathcal{B}}{\partial x}, \quad \mathbf{L}(x)\equiv0;
\end{equation*}
if $\mathbf{U}(x_{t}):=\Upsilon(\mathcal{V}(x_{t}), \mathcal{B}(x_{t}))$, then $\mathcal{V}(\phi)=\mathcal{V}_{1}(\phi(0))+\mathcal{V}_{2}(\phi)$, $\mathcal{B}(\phi)=\mathcal{B}_{1}(\phi(0))+\mathcal{B}_{2}(\phi)$ and
\begin{align*}
\mathbf{H}(\phi)&=\frac{\partial\Upsilon}{\partial\mathcal{V}}\frac{\partial\mathcal{V}_{1}}{\partial\phi(0)}+\frac{\partial\Upsilon}{\partial\mathcal{B}}\frac{\partial\mathcal{B}_{1}}{\partial\phi(0)}, \\ \mathbf{L}(\phi)&=\frac{\partial\Upsilon}{\partial\mathcal{V}}D^{+}\mathcal{V}_{2}(\phi)+\frac{\partial\Upsilon}{\partial\mathcal{B}}D^{+}\mathcal{B}_{2}(\phi).
\end{align*}
In this section, the following assumption is made, which is called the transversality condition \citep{Sira1999general} and used to avoid $g(\phi)$ to be orthogonal to $\mathbf{H}(\phi)$.

\begin{assumption}
\label{asp-1}
For all $\phi\in\mathcal{C}([-\Delta, 0], \mathbb{R}^{n})$, $\mathbf{G}(\phi)\neq0$.
\end{assumption}

Assumption \ref{asp-1} ensures that $g(\phi)$ is not tangential to the level set of the sliding surface functional $\mathbf{U}(\phi)$. If Assumption \ref{asp-1} does not hold, then higher-order sliding surface functionals can be introduced \citep{Shtessel2014sliding, Oguchi2006sliding}, and then the following analysis can be proceeded similarly. This will be further discussed in Section \ref{subsec-discussion}.

With Assumption \ref{asp-1}, we define the matrix functional:
\begin{align}
\label{eqn-45}
\mathbf{M}(\phi):=I-\frac{g(\phi)\mathbf{G}^{\top}(\phi)\mathbf{H}(\phi)}{\|\mathbf{G}(\phi)\|^{2}}.
\end{align}
The next proposition shows the properties of $\mathbf{M}(\phi)$, which extends the results in \citep{Sira1999general} into the time-delay case and lays a foundation for the control design.

\begin{proposition}
\label{prop-2}
The following statements are valid:
\begin{align}
\label{eqn-46}
\mathbf{M}(\phi)f(\phi)&=\mathbf{J}_{1}(\phi)\mathbf{H}^{\top}(\phi), \\
\label{eqn-47}
(I-\mathbf{M}(\phi))f(\phi)&=\mathbf{J}_{2}(\phi)\mathbf{H}^{\top}(\phi)+\mathbf{J}_{3}(\phi)\mathbf{H}^{\top}(\phi),
\end{align}
where $\mathbf{J}_{1}(\phi)=-2\mathbf{J}_{2}(\phi)$, and
\begin{align*}
\mathbf{J}_{2}(\phi)&=\frac{g(\phi)\mathbf{G}^{\top}(\phi)f^{\top}(\phi)-f(\phi)\mathbf{G}(\phi)g^{\top}(\phi)}{2\|\mathbf{G}(\phi)\|^{2}}, \\
\mathbf{J}_{3}(\phi)&=\frac{f(\phi)\mathbf{G}(\phi)g^{\top}(\phi)+g(\phi)\mathbf{G}^{\top}(\phi)f^{\top}(\phi)}{2\|\mathbf{G}(\phi)\|^{2}}.
\end{align*}
\end{proposition}

\begin{proof}
Since $\mathbf{H}(\phi)f(\phi)\in\mathbb{R}$ and $\mathbf{G}(\phi)=\mathbf{H}(\phi)g(\phi)\in\mathbb{R}^{1\times m}$, we have that $\mathbf{H}(\phi)f(\phi)=f^{\top}(\phi)\mathbf{H}^{\top}(\phi)$ and $\|\mathbf{G}(\phi)\|^{2}=\mathbf{G}(\phi)\mathbf{G}^{\top}(\phi)$. Define $\mathbf{J}(\phi):=f(\phi)\mathbf{G}(\phi)g^{\top}(\phi)$, and we have from \eqref{eqn-45} that
\begin{align*}
\mathbf{M}(\phi)f(\phi)&=f(\phi)-\frac{g(\phi)\mathbf{G}^{\top}(\phi)\mathbf{H}(\phi)f(\phi)}{\|\mathbf{G}(\phi)\|^{2}} \\
&=\frac{f(\phi)\|\mathbf{G}(\phi)\|^{2}-g(\phi)\mathbf{G}^{\top}(\phi)\mathbf{H}(\phi)f(\phi)}{\|\mathbf{G}(\phi)\|^{2}} \\
&=\frac{f(\phi)\mathbf{G}(\phi)\mathbf{G}^{\top}(\phi)-{\mathbf{J}}^{\top}(\phi)\mathbf{H}^{\top}(\phi)}{\|\mathbf{G}(\phi)\|^{2}} \\
&=\frac{({\mathbf{J}}(\phi)-{\mathbf{J}}^{\top}(\phi))\mathbf{H}^{\top}(\phi)}{\|\mathbf{G}(\phi)\|^{2}}.
\end{align*}
Let $\mathbf{J}_{1}(\phi):=({\mathbf{J}}(\phi)-{\mathbf{J}}^{\top}(\phi))/\|\mathbf{G}(\phi)\|^{2}$, and thus \eqref{eqn-46} holds.
\begin{align*}
(I-\mathbf{M}(\phi))f(\phi)&=\frac{g(\phi)\mathbf{G}^{\top}(\phi)\mathbf{H}(\phi)f(\phi)}{\|\mathbf{G}(\phi)\|^{2}}
=\frac{{\mathbf{J}}^{\top}(\phi)\mathbf{H}^{\top}(\phi)}{\|\mathbf{G}(\phi)\|^{2}} \\
&=\frac{{\mathbf{J}}^{\top}(\phi)-{\mathbf{J}}(\phi)}{2\|\mathbf{G}(\phi)\|^{2}}\mathbf{H}^{\top}(\phi)
+\frac{{\mathbf{J}}^{\top}(\phi)+{\mathbf{J}}(\phi)}{2\|\mathbf{G}(\phi)\|^{2}}\mathbf{H}^{\top}(\phi) \\
&=\mathbf{J}_{2}(\phi)\mathbf{H}^{\top}(\phi)+\mathbf{J}_{3}(\phi)\mathbf{H}^{\top}(\phi).
\end{align*}
Hence, we conclude that \eqref{eqn-47} holds, and $\mathbf{J}_{1}(\phi)=-2\mathbf{J}_{2}(\phi)$.
\hfill$\blacksquare$
\end{proof}

From Proposition \ref{prop-2}, the functionals $\mathbf{J}_{1}(\phi), \mathbf{J}_{2}(\phi)$ are skew-symmetric and the functional $\mathbf{J}_{3}(\phi)$ is symmetric. In addition,
\begin{align*}
\mathbf{H}(\phi)\mathbf{J}_{1}(\phi)\mathbf{H}^{\top}(\phi)&=-2\mathbf{H}(\phi)\mathbf{J}_{2}(\phi)\mathbf{H}^{\top}(\phi)\equiv0, \\
f(\phi)&=(\mathbf{J}_{3}(\phi)-\mathbf{J}_{2}(\phi))\mathbf{H}^{\top}(\phi).
\end{align*}
Therefore, the system \eqref{eqn-1} can be rewritten to be related to the sliding surface functional \eqref{eqn-44}. That is, the dynamics in \eqref{eqn-1} can written as $\dot{x}(t)=-\mathbf{J}_{2}(x_{t})\mathbf{H}^{\top}(x_{t})+\mathbf{J}_{3}(x_{t})\mathbf{H}^{\top}(x_{t})+g(x_{t})u$.

\subsection{Feedback Control Design}

With the sliding surface functional \eqref{eqn-44}, we next address the controller design. In the ideal sliding motion case, the system state is to satisfy the manifold invariant condition:
\begin{align}
\label{eqn-48}
\mathbf{U}(\phi)=0,
\end{align}
which can be verified via the functional $\mathbf{W}(\phi):=0.5\mathbf{U}^{2}(\phi)$. Specifically, the time derivative of $\mathbf{W}(\phi)$ is given by
\begin{align}
\label{eqn-49}
D^{+}\mathbf{W}(\phi)&=\mathbf{U}(\phi)D^{+}\mathbf{U}(\phi)\nonumber\\
&=\mathbf{U}(\phi)(-\mathbf{H}(\phi)\mathbf{J}_{2}(\phi)\mathbf{H}^{\top}(\phi)+\mathbf{H}(\phi)\mathbf{J}_{3}(\phi)\mathbf{H}^{\top}(\phi) \nonumber\\
&\quad+\mathbf{L}(\phi)+\mathbf{G}(\phi)u) \nonumber\\
&=\mathbf{U}(\phi)(\mathbf{H}(\phi)\mathbf{J}_{3}(\phi)\mathbf{H}^{\top}(\phi)+\mathbf{L}(\phi)+\mathbf{G}(\phi)u).
\end{align}
By the invariance condition $D^{+}\mathbf{W}(\phi)=0$, the equivalent control law is derived as
\begin{align*}
u_{\textsf{e}}(\phi)=\frac{\mathbf{G}^{\top}(\phi)(\mathbf{H}(\phi)\mathbf{J}_{3}(\phi)\mathbf{H}^{\top}(\phi)+\mathbf{L}(\phi))}{-\|\mathbf{G}(\phi)\|^{2}}.
\end{align*}
The equivalent control law $u_{\textsf{e}}$ neutralizes all working forces which locally affect the magnitude of the sliding surface coordinate. Since the state trajectory of the system \eqref{eqn-1} may move into the sublevel and superlevel sets of the sliding surface functional, the applied controller is of the following form:
\begin{align}
\label{eqn-50}
u(\phi)=\frac{\mathbf{G}^{\top}(\phi)(\mathbf{H}(\phi)\mathbf{J}_{3}(\phi)\mathbf{H}^{\top}(\phi)+\mathbf{L}(\phi)+\mathbf{K}(\phi))}{-\|\mathbf{G}(\phi)\|^{2}},
\end{align}
where $\mathbf{K}(\phi)>0$ is the additional item to be designed based on the applied sliding surface functional $\mathbf{U}(\phi)$. Obviously, $u(\phi)=u_{\textsf{e}}(\phi)-\|\mathbf{G}(\phi)\|^{-2}\mathbf{G}^{\top}(\phi)\mathbf{K}(\phi)$.

From all above discussion, we have the following theorem, which shows the satisfaction of both stabilization and safety objectives via the sliding surface functional $\mathbf{U}(\phi)$ in \eqref{eqn-44}.

\begin{theorem}
\label{thm-8}
Consider the system \eqref{eqn-1} with the safe set $\mathbb{C}\subset\mathbb{R}^{n}$ and the initial state $\xi\in\mathcal{C}([-\Delta, 0], \inte(\mathbb{C}))$. If Assumption \ref{asp-1} holds, and the sliding surface functional $\mathbf{U}$ in \eqref{eqn-44} is such that
\begin{align}
\label{eqn-51}
&|\mathbf{U}(\phi(\theta))|\geq|\mathbf{U}(\xi(\theta))|, \quad \forall \phi\in\mathcal{C}([-\Delta, 0], \partial\mathbb{C}), \\
\label{eqn-52}
&\mathbb{S}:=\{\phi\in\mathcal{C}([-\Delta, 0], \mathbb{C}): \mathbf{U}(\phi)=0\}\subset\mathcal{C}([-\Delta, 0], \inte(\mathbb{C})),
\end{align}
then the stabilization and safety objectives can be achieved simultaneously via the controller \eqref{eqn-50} with
\begin{equation}
\label{eqn-53}
\mathbf{K}(\phi):=\mathrm{K}\sign(\mathbf{U}(\phi)),
\end{equation}
where $\mathrm{K}>0$ is constant.
\end{theorem}

\begin{proof}
From \eqref{eqn-49}, \eqref{eqn-50} and \eqref{eqn-53}, we have
\begin{align*}
D^{+}\mathbf{W}(\phi)&=\mathbf{U}(\phi)(\mathbf{H}(\phi)\mathbf{J}_{3}(\phi)\mathbf{H}^{\top}(\phi)+\mathbf{L}(\phi)+\mathbf{G}(\phi)u)  \nonumber  \\
&=-\mathbf{U}(\phi)\mathrm{K}\sign(\mathbf{U}(\phi))  \nonumber  \\
&=-\mathrm{K}\sqrt{2\mathbf{W}(\phi)}.
\end{align*}
From \citep[Theorem 2.2]{Pepe2014stabilization}, the functional $\mathbf{W}(\phi)$ converges to the origin with the increase of time, which implies the satisfaction of the stabilization objective. In addition, we can see from \eqref{eqn-52} that the manifold invariant condition is included in the safe set, and from \eqref{eqn-51} that the evolution of $\mathbf{W}(\phi)$ is in the safe set. Hence, the safety objective is guaranteed.
\hfill $\blacksquare$
\end{proof}

Theorem \ref{thm-8} involves a standard method to design $\mathbf{K}(\phi)$ in \eqref{eqn-50}. Here, we emphasize that the choice of $\mathbf{K}(\phi)$ is not unique. In particular, due to the function $\sign$ in \eqref{eqn-53}, the controller \eqref{eqn-50} is not continuous. To deal with this issue, we can replace the function $\sign$ by a sigmoid function such that $\mathbf{K}(\phi):=\mathrm{K}\mathbf{U}(\phi)/(\|\mathbf{U}(\phi)\|+\varepsilon)$ with arbitrarily small $\varepsilon>0$ is continuous; see also Section 1.2.1 in \citep{Shtessel2014sliding}. The condition \eqref{eqn-51} implies that $\mathbf{W}(\phi)$ will not converge to the boundary of the safe set, and the condition \eqref{eqn-52} ensures the sliding surface to be included in the safe set. Hence, the safety objective is guaranteed via \eqref{eqn-51}-\eqref{eqn-52}, which can be strengthened to
\begin{align*}
&|\mathbf{U}(\phi(\theta))|\geq|\mathbf{U}(\xi(\theta))|, \quad \forall \phi\in\mathcal{C}([-\Delta, 0], \partial\mathbb{C}_{1}), \\
&\mathbb{S}=\{\phi\in\mathcal{C}([-\Delta, 0], \mathbb{C}): \mathbf{U}(\phi)=0\}\subset\mathcal{C}([-\Delta, 0], \mathbb{C}_{1}),
\end{align*}
where $\mathbb{C}_{1}=\mathbb{C}-\mathbf{B}(\varepsilon)$ and $\varepsilon>0$ is arbitrarily small. In this case, the safety objective can be guaranteed under some possible chattering phenomena. Next, we show how to construct this sliding surface functional \eqref{eqn-44} such that \eqref{eqn-51}-\eqref{eqn-52} are satisfied under the controller \eqref{eqn-50}. For this purpose, two types of sliding surface functionals are constructed, and sufficient conditions are established in the following propositions.

\begin{proposition}
\label{prop-3}
Consider the system \eqref{eqn-1} with the safe set $\mathbb{C}\subset\mathbb{R}^{n}$ and the initial state $\xi\in\mathcal{C}([-\Delta, 0], \inte(\mathbb{C}))$. If
\begin{enumerate}[(i)]
  \item the control Lyapunov functional is $\mathcal{V}(\phi)$,
  \item the Z-type control barrier functional is $\mathcal{B}(\phi)$,
  \item the sliding surface functional is
  \begin{align}
  \label{eqn-54}
  \mathbf{U}(\phi)&:=\mathbf{a}\mathcal{V}(\phi)-\mathbf{b}\mathcal{B}(\phi)+\mathbf{c},
  \end{align}
  with $\mathbf{a}, \mathbf{b}, \mathbf{c}>0$ satisfying
  \begin{equation}
  \label{eqn-55}
  0\leq\|\mathbf{U}(\xi)\|<\mathbf{c},
  \end{equation}
\end{enumerate}
then the controller \eqref{eqn-50} guarantees simultaneously the stabilization and safety objectives of the system \eqref{eqn-1}.
\end{proposition}

\begin{proof}
From Theorem \ref{thm-8}, the sliding surface converges to the origin under the controller \eqref{eqn-50}, and thus the stabilization objective is guaranteed. In the following, we just show the satisfaction of the safety objective.

Since $D^{+}\mathbf{W}(\phi)<0$, we have $|\mathbf{U}(x_{t}(\theta))|\leq|\mathbf{U}(x_{0}(\theta))|$ for all $t>0$ and $\theta\in[-\Delta, 0]$. In the following, we consider two cases. For the case $\mathbf{U}(x_{t}(\theta))>0$,
\begin{equation*}
0<\mathbf{a}\mathcal{V}(x_{t}(\theta))-\mathbf{b}\mathcal{B}(x_{t}(\theta))+\mathbf{c}\leq|\mathbf{U}(x_{0}(\theta))|.
\end{equation*}
From \eqref{eqn-55}, if $\mathbf{U}(x_{0}(\theta))\geq0$ for all $\theta\in[-\Delta, 0]$, then $\mathbf{a}\mathcal{V}(x_{0}(\theta))-\mathbf{b}\mathcal{B}(x_{0}(\theta))<0$, and
\begin{align}
\label{eqn-56}
\mathbf{b}\mathcal{B}(x_{t}(\theta))&\geq\mathbf{a}\mathcal{V}(x_{t}(\theta))-(\mathbf{a}\mathcal{V}(x_{0}(\theta))-\mathbf{b}\mathcal{B}(x_{0}(\theta))).
\end{align}
Since $\mathbf{a}\mathcal{V}(x_{t})\geq0$ for all $t\in\mathbb{R}^{+}$, we have from \eqref{eqn-56} that $\mathcal{B}(x_{t})>0$ for all $t\in\mathbb{R}^{+}$. If $\mathbf{U}(x_{0}(\theta))<0$ for all $\theta\in[-\Delta, 0]$, then $|\mathbf{U}(x_{0}(\theta))|=-\mathbf{a}\mathcal{V}(x_{0}(\theta))+\mathbf{b}\mathcal{B}(x_{0}(\theta))-\mathbf{c}$, and
\begin{align*}
\mathbf{b}\mathcal{B}(x_{t}(\theta))&\geq\mathbf{a}\mathcal{V}(x_{t}(\theta))+2\mathbf{c}-(\mathbf{a}\mathcal{V}(x_{0}(\theta))-\mathbf{b}\mathcal{B}(x_{0}(\theta))),
\end{align*}
where $2\mathbf{c}-(\mathbf{a}\mathcal{V}(x_{0}(\theta))-\mathbf{b}\mathcal{B}(x_{0}(\theta)))\geq0$ holds from \eqref{eqn-55}. Hence, $\mathcal{B}(x_{t}(\theta))>0$ for all $t\in\mathbb{R}^{+}$ and $\theta\in[-\Delta, 0]$.

For the case $\mathbf{U}(x_{t})<0$, if $\mathbf{U}(x_{0}(\theta))\geq0$ for all $\theta\in[-\Delta, 0]$, then $-\mathbf{a}\mathcal{V}(x_{0}(\theta))+\mathbf{b}\mathcal{B}(x_{0}(\theta))-\mathbf{c}
<\mathbf{a}\mathcal{V}(x_{t}(\theta))-\mathbf{b}\mathcal{B}(x_{t}(\theta))+\mathbf{c}<0$, which implies from \eqref{eqn-55} that
\begin{align*}
\mathbf{b}\mathcal{B}(x_{t}(\theta))&>\mathbf{a}\mathcal{V}(x_{t}(\theta))+\mathbf{c}, \\
\mathbf{b}\mathcal{B}(x_{t}(\theta))&<\mathbf{a}\mathcal{V}(x_{t}(\theta))+2\mathbf{c}-(\mathbf{a}\mathcal{V}(x_{0}(\theta))-\mathbf{b}\mathcal{B}(x_{0}(\theta))).
\end{align*}
Hence, $\mathcal{B}(x_{t}(\theta))>0$ for all $t\in\mathbb{R}^{+}$ and $\theta\in[-\Delta, 0]$. If $\mathbf{U}(x_{0}(\theta))<0$ for all $\theta\in[-\Delta, 0]$, then we follow the similar mechanism to derive $\mathcal{B}(x_{t}(\theta))>0$ for all $t\in\mathbb{R}^{+}$ and $\theta\in[-\Delta, 0]$. That is, the safety objective is guaranteed.
\hfill$\blacksquare$
\end{proof}

In Proposition \ref{prop-3}, the functional \eqref{eqn-54} is a simple and linear combination of the proposed CLRF (or CLKF) and Z-CBRF (or Z-CBKF). In this case, the condition \eqref{eqn-55} is required to be satisfied, which depends only on the initial state. Hence, given the initial state, the choice of the variables $\mathbf{a}, \mathbf{b}, \mathbf{c}>0$ is constrained. Once $\mathbf{a}, \mathbf{b}, \mathbf{c}>0$ are fixed, the condition \eqref{eqn-55} in turn constrains the initial state starting from which the stabilization and safety objectives can be achieved simultaneously.

\begin{proposition}
\label{prop-4}
Consider the system \eqref{eqn-1} with the safe set $\mathbb{C}\subset\mathbb{R}^{n}$ and the initial state $\xi\in\mathcal{C}([-\Delta, 0], \inte(\mathbb{C}))$. If
\begin{enumerate}[(i)]
  \item the control Lyapunov functional is $\mathcal{V}(\phi)$,
  \item the control barrier functional is $\mathcal{B}(\phi)$,
  \item the sliding surface functional is
  \begin{align}
  \label{eqn-57}
  \mathbf{U}(\phi)&:=\alpha(\mathcal{V}(\phi))+\beta(\mathcal{B}(\phi)),
  \end{align}
  with the functions $\alpha, \beta: \mathbb{R}\rightarrow\mathbb{R}^{+}$ satisfying
  \begin{equation}
  \label{eqn-58}
  \mathbf{U}(\phi)\geq\mathbf{U}(\xi), \quad \forall \phi\in\mathcal{C}([-\Delta, 0], \partial\mathbb{C}),
  \end{equation}
\end{enumerate}
then the controller \eqref{eqn-50} guarantees simultaneously the stabilization and safety objectives of the system \eqref{eqn-1}.
\end{proposition}

\begin{proof}
From Theorem \ref{thm-8}, the stabilization objective is guaranteed via the controller \eqref{eqn-50}, and we next show the satisfaction of the safety objective. From the construction of $\mathbf{U}(\phi)$, $\mathbf{U}(x_{t})\geq0$ for all $t>0$. Since $D^{+}\mathbf{W}(\phi)<0$, we have $\mathbf{U}(x_{t})<\mathbf{U}(x_{0})$ for all $t>0$. From \eqref{eqn-58}, $\mathbf{U}(\phi)\geq\mathbf{U}(\xi)$ for all $\phi\in\mathcal{C}([-\Delta, 0],\partial\mathbb{C})$, the state trajectory cannot reach the boundary of the safe set. Note that $\xi\in\mathcal{C}([-\Delta, 0], \inte(\mathbb{C}))$, and thus the state trajectory will stay in the safe set, which in turn ensures the safety objective.
\hfill $\blacksquare$
\end{proof}

\begin{remark}
\label{rmk-3}
The condition \eqref{eqn-58} can be strengthened to be satisfied for $\phi\in\mathcal{C}([-\Delta, 0], \partial\mathbb{C}_{1})$ with $\mathbb{C}_{1}=\mathbb{C}-\mathbf{B}(\varepsilon)$ and arbitrarily small $\varepsilon>0$. That is, the condition \eqref{eqn-58} is satisfied when the state is closed to $\partial\mathbb{C}$ instead of reaches $\partial\mathbb{C}$. In such case, the reach of the boundary of the safe set is avoided to attenuate the possibility of unsafety due to chattering phenomena.
\hfill $\square$
\end{remark}

Different from Proposition \ref{prop-3}, no constraints are imposed on the type of the CBF in Proposition \ref{prop-4}. The sliding surface functional \eqref{eqn-57} is extensively applied in many existing works \citep{Ames2016control, Panagou2015distributed}, and \eqref{eqn-58} is not strict and can be achieved by the construction. To show this, we illustrate the choices of $\alpha$ and $\beta$ in \eqref{eqn-57}. First, $\alpha(\mathcal{V})$ can be set as $\mathcal{V}$ directly. Another common construction of $\alpha(\mathcal{V})$ is based on the hyperbolic function \citep{Badawy2009small} and given as $\alpha(\mathcal{V}(\phi)):=a\sqrt{\mathcal{V}(\phi)+b}-b$, where $a, b>0$. Second, if $\mathcal{B}$ is an R-CBKF satisfying $\mathcal{B}(\phi)\geq\mathcal{B}(\xi)$ for all $\phi\in\mathcal{C}([-\Delta, 0], \partial\mathbb{C})$, then $\beta(\mathcal{B})$ can be set as $\mathcal{B}$ directly. For instance, $\mathcal{B}(\phi)=\ln(1+c/h(\phi))$ with $c>0$ ensures that $\mathcal{B}(\phi)\rightarrow\infty$ as $\phi\rightarrow\mathcal{C}([-\Delta, 0], \partial\mathbb{C})$. We can limit the effects of the CBF by the following construction: $\beta(\mathcal{B}(\phi))=c\cdot\sgn(\beta_{1}(h(\phi)))\beta^{2}_{1}(h(\phi))h^{-2}(\phi)$, where $c>0$ and $\beta_{1}: \mathbb{R}^{+}\rightarrow\mathbb{R}$ is such that $\mathbb{S}_{1}:=\{\phi\in\mathcal{C}([-\Delta, 0], \mathbb{R}^{n}): \beta_{1}(h(\phi))=0\}\subset\mathcal{C}([-\Delta, 0], \inte(\mathbb{C}))$. That is, the set $\mathbb{S}_{1}$ is a contraction of the safe set. Only when the state trajectory moves into the set $\mathcal{C}([-\Delta, 0], \inte(\mathbb{C}))\setminus\mathbb{S}_{1}$ does the CBF affects the controller design to guarantee the safety objective. In particular, $\beta_{1}(h(\phi))$ can be set as $h(\phi)-\varepsilon$ simply with $\varepsilon>0$ determining the distance between the sets $\mathbb{S}$ and $\mathbb{S}_{1}$; see Section \ref{subsec-MSyn}.

\begin{remark}
\label{rmk-4}
Since whether $\mathcal{V}$ and $\mathcal{B}$ are the Razumikhin-type or Krasovskii-type is not specified in \eqref{eqn-44}, the results derived in this section are applied to these two types. An advantage of these results lies in the potential combination between different types of CLFs and CBFs. That is, CLRF (or CLKF) and CBKF (or CBRF) can be merged to study the stabilization and safety objectives simultaneously; see Section \ref{subsec-MSyn}.
\hfill $\square$
\end{remark}

\subsection{Further Discussion}
\label{subsec-discussion}

Since SMC techniques are applied here, chattering phenomena may exist, which mean the oscillations with finite frequency and amplitude caused by the switching around the sliding surface \citep{Sira1999general}. In addition, Assumption \ref{asp-1} results in some conservatism since this assumption is not necessarily always valid.

The chattering phenomena occur due to the sign function in \eqref{eqn-53}. To attenuate the chattering phenomena, we can redesign the controller \eqref{eqn-50} by adjusting the item $\mathbf{K}(\phi)$. In particular, for Razumikhin-type control functions,
\begin{equation}
\label{eqn-59}
\mathbf{K}(\phi)=\mathrm{K}_{1}\mathbf{U}(\phi(0))-\frac{\alpha(\|\mathbf{U}(\phi)\|^{2})}{\mathbf{U}(\phi(0))},
\end{equation}
where $\mathrm{K}_{1}>0$ and $\alpha\in\mathcal{K}$ satisfying $0<\alpha(v)<\mathrm{K}_{1}v$ for all $v>0$. For Krasovskii-type control functionals,
\begin{equation}
\label{eqn-60}
\mathbf{K}(\phi)=\mathrm{K}_{1}\mathbf{U}(\phi), \quad \mathrm{K}_{1}>0.
\end{equation}
With \eqref{eqn-59}-\eqref{eqn-60}, the convergence of the sliding surface functional can be achieved via Corollary 1 in \citep{Pepe2021nonlinear} and Theorem 2.2 in \citep{Pepe2014stabilization}, respectively. Since the sign function is avoided in these settings, the chattering phenomena are attenuated or avoided; see Section \ref{subsec-CCC}.

\begin{figure*}
\centering
\begin{tabular}{cccc}
\includegraphics[width=0.47\columnwidth]{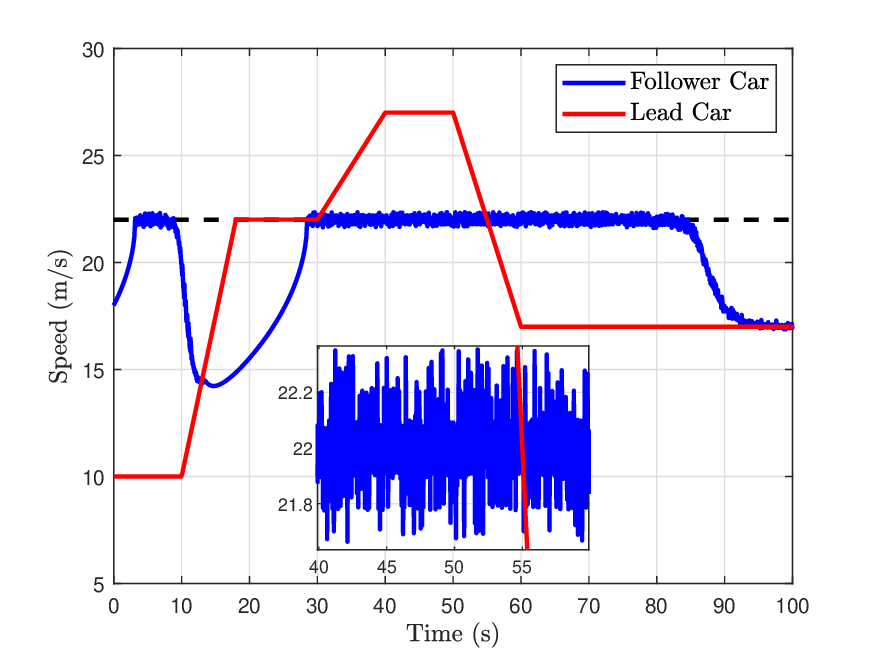} \vspace{-1pt} &
\includegraphics[width=0.47\columnwidth]{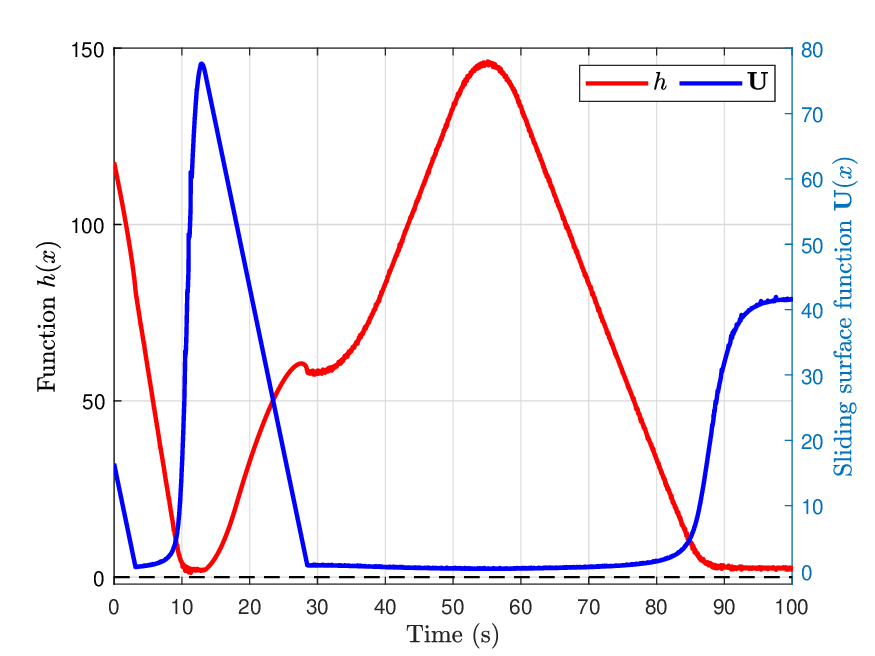} \vspace{-1pt} &
\includegraphics[width=0.47\columnwidth]{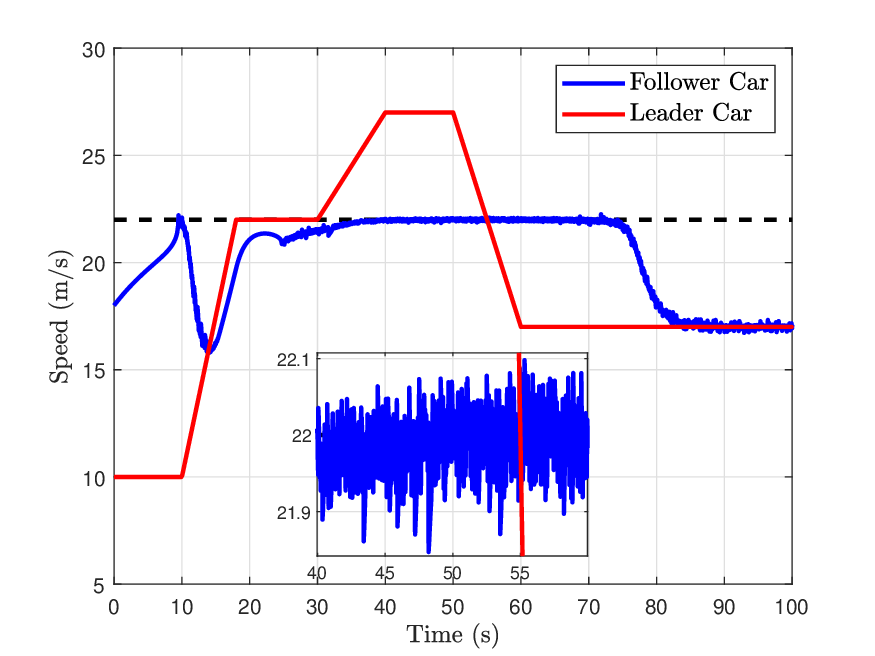} \vspace{-1pt} &
\includegraphics[width=0.47\columnwidth]{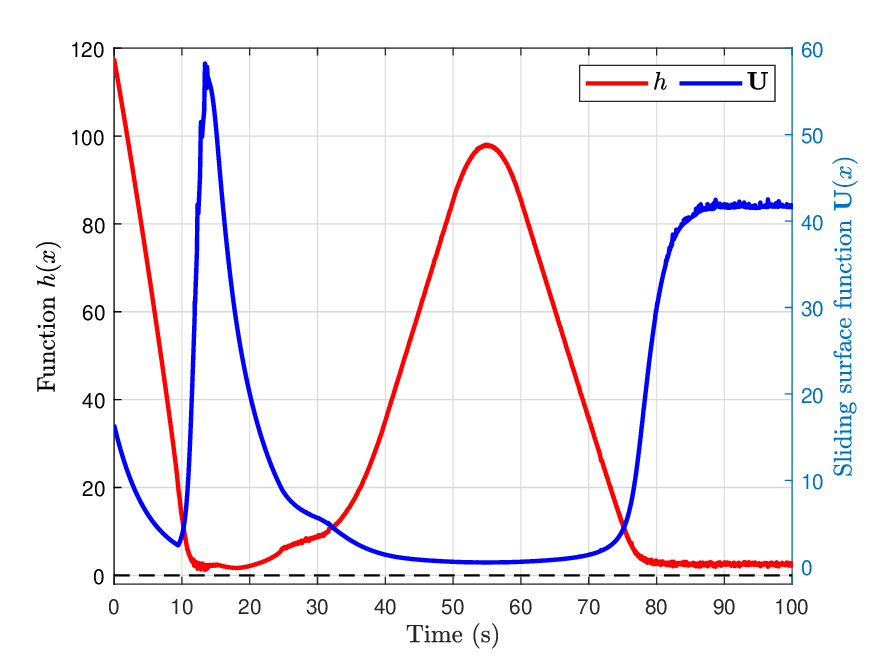} \vspace{-1pt} \\
\tiny{(a)} & \tiny{(b)} & \tiny{(c)} & \tiny{(d)}
\end{tabular}
\caption{\footnotesize Simulation of the CCC problem under two different controllers and $\Delta=0.2$. (a)-(b): the controller \eqref{eqn-50} with $\mathbf{K}(\phi)=5\sign(\mathbf{U}(\phi(0)))$. (c)-(d): the controller \eqref{eqn-50} with $\mathbf{K}(\phi)=2.2\mathbf{U}(\phi(0))-2\|\mathbf{U}(\phi)\|^{2}/\mathbf{U}(\phi(0))$.}
\label{fig-1}
\end{figure*}

\begin{figure*}
\centering
\begin{tabular}{cccc}
\includegraphics[width=0.47\columnwidth]{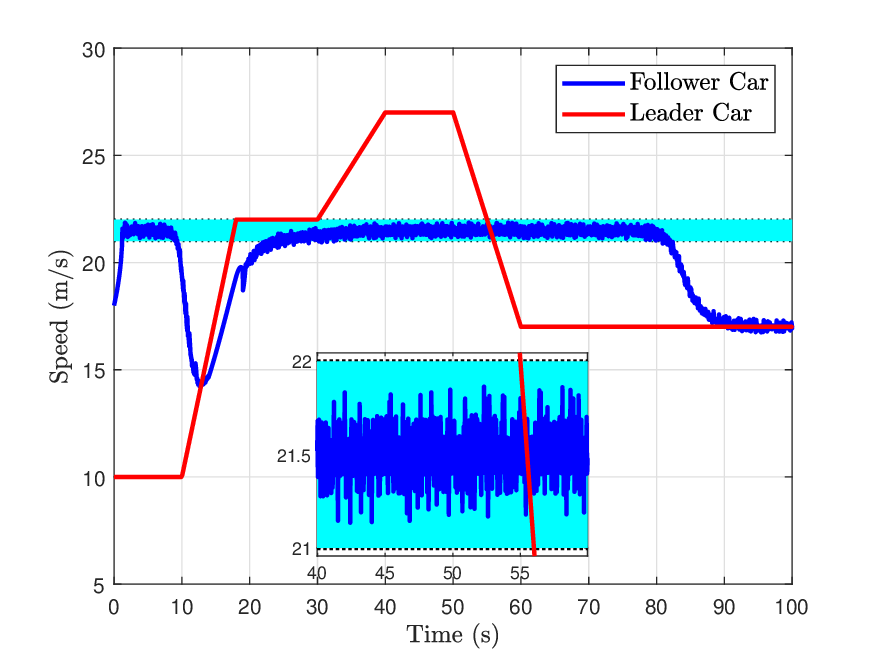}\vspace{-1pt} &
\includegraphics[width=0.47\columnwidth]{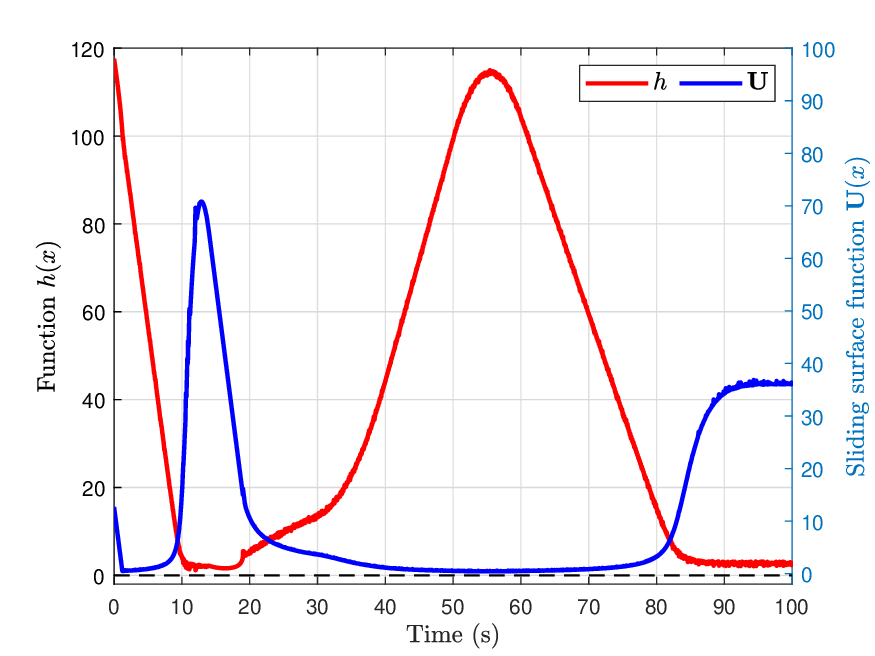}\vspace{-1pt} &
\includegraphics[width=0.47\columnwidth]{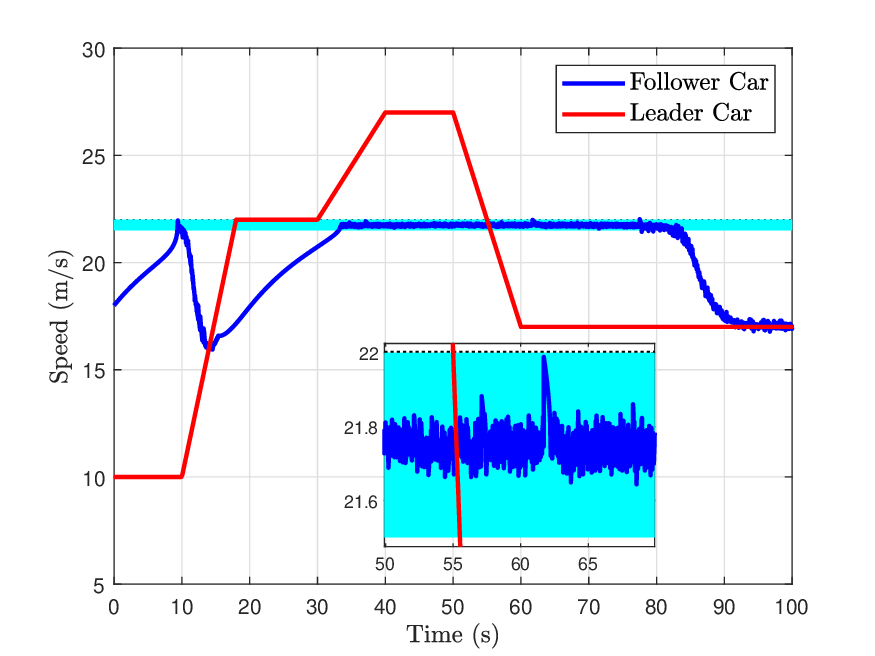}\vspace{-1pt} &
\includegraphics[width=0.47\columnwidth]{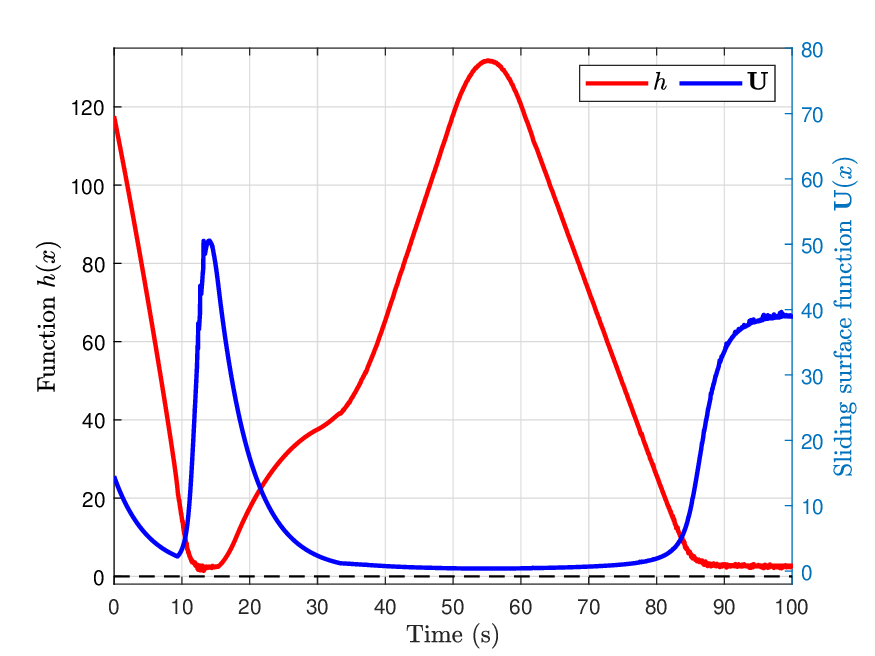}\vspace{-1pt} \\
\tiny{(a)} & \tiny{(b)} & \tiny{(c)} & \tiny{(d)}
\end{tabular}
\caption{\footnotesize  Simulation of the CCC problem under two different controllers and $\Delta=0.5$. The cyan region is the desired velocity band for the follower car. (a)-(b): the controller \eqref{eqn-50} with $\mathbf{K}(\phi)=10\sign(\mathbf{U}(\phi(0)))$ and the desired band $[21, 22]$. (c)-(d): the controller \eqref{eqn-50} with $\mathbf{K}(\phi)=2.2\mathbf{U}(\phi(0))-2\|\mathbf{U}(\phi)\|^{2}/\mathbf{U}(\phi(0))$ and the desired band $[21.6, 22]$.}
\label{fig-2}
\end{figure*}

If the transversality condition, that is, Assumption \ref{asp-1}, does not hold, then higher-order sliding surface functionals \citep{Shtessel2014sliding, Oguchi2006sliding} can be applied to guarantee the properties of the above controller on the one hand and to attenuate the chattering phenomena on the other hand. To show this, we take the relative degree 2 case as an example, because the second-order SMC is extensively applied to attenuate the chattering phenomena \citep{Shtessel2014sliding}. In this case, the sliding surface functional \eqref{eqn-44} satisfies
\begin{align}
\label{eqn-61}
L_{g}\mathbf{U}(\phi)=0, \quad L_{g}L_{f}\mathbf{U}(\phi)\neq0.
\end{align}
Hence, the manifold invariant condition \eqref{eqn-48} is changed to
\begin{align*}
\mathbf{U}(\phi)=0, \quad  D^{+}\mathbf{U}(\phi)=0.
\end{align*}
Furthermore, the sliding surface functional can be modified slightly as
\begin{equation*}
\overbar{\mathbf{U}}(\phi):=\mathbf{a}\mathbf{U}(\phi)+\mathbf{b}L_{f}\mathbf{U}(\phi),
\end{equation*}
where $\mathbf{a}, \mathbf{b}$ are nonzero constants. Hence, we have from \eqref{eqn-61} that $L_{g}\overbar{\mathbf{U}}(\phi)=\mathbf{b}L_{g}L_{f}\mathbf{U}(\phi)\neq0$, and the transversality condition is avoided. With the sliding surface functional $\overbar{\mathbf{U}}(\phi)$, the analysis in this section can be proceeded similarly to derive the feedback controller. Following a similar mechanism, the sliding surface functional can be defined for the higher-order case, and similar results can be obtained; see, e.g., \cite{Oguchi2006sliding}.

\section{Application to Multi-Agent Systems}
\label{sec-examples}

In this section, two numerical examples are presented to illustrate the results derived in previous sections. The computation is executed using MATLAB R2020a on a Dell laptop with a 16GB RAM and an Intel i7-10610U processor (1.80GHz).

\subsection{Connected Cruise Control}
\label{subsec-CCC}

Based on \citep{Hsia1990robot} using time-delay control techniques and \citep{Jin2017delay} involving driver reaction and communication delays, we consider the following continuous-time car-following model, where both leader and follower cars are modeled as a point-masses moving in the straight line:
\begin{align}
\label{eqn-62}
\dot{x}(t)=\begin{bmatrix}F(x(t-\tau))-F(x(t)) \\ a \\ x_{2}(t)-x_{1}(t)\end{bmatrix}+\begin{bmatrix} 1 \\ 0 \\ 0\end{bmatrix}u,
\end{align}
where $x=(x_{1}, x_{2}, x_{3})\in\mathbb{R}^{3}$, $x_{1}\in\mathbb{R}$ and $x_{2}\in\mathbb{R}$ are respectively the velocities of the follower and leader cars, $x_{3}\in\mathbb{R}$ is the distance between these two cars, and $u\in\mathbb{R}$ is the wheel force to be designed as the control input of the follower car. In \eqref{eqn-62}, $F\in\mathbb{R}$ is the total sum of the nonlinear dynamics of car, drag, frictions and disturbances; $\tau\in\mathbb{R}$ is the time delay (in s) due to delayed driver reactions or time-delay control techniques; and $a\in\mathbb{R}$ is the acceleration of the leader car (in m/s$^{2}$). Both $\tau$ and $a$ are bounded. Following the delay-free case \citep{Ames2016control}, $F(x)$ is set as $(a_{0}+a_{1}x_{1}+a_{2}x^{2}_{1})/M$ with $M>0$ being the mass of the follower car (in kg) and constants $a_{0}, a_{1}, a_{2}\in\mathbb{R}$ determined empirically.

In terms of connected cruise control, the follower car is expected to follow the leader car with a desired speed safely, which can be divided into two goals. The first goal is to achieve the desired velocity $v_{\textsf{d}}\in\mathbb{R}$. For this purpose, we introduce the candidate control Lyapunov function:
\begin{align}
\label{eqn-63}
V_{\rzmk}(x):=(x_{1}-v_{\textsf{d}})^{2}.
\end{align}
Hence, the first goal is transformed into the control design to guarantee the convergence of the Lyapunov function $V_{\rzmk}(x)$. From Definition \ref{def-3}, $V_{\rzmk}$ is a CLRF if the condition \eqref{eqn-8} holds. The second goal is to ensure the follower car moves safely, that is, the distance between the follower and leader cars are nonnegative. We emphasize that the second goal has priority over the first goal. To transform the second goal into a mathematical expression, we introduce the function $h_{\rzmk}(x)=x_{3}-x_{1}\mathbf{t}$ with the desired time headway $\mathbf{t}=1.8$s. That is, the safety objective is satisfied only when $h_{\rzmk}(x)\geq0$. With the function $h_{\rzmk}(x)$, the candidate control barrier function is given by
\begin{align}
\label{eqn-64}
B_{\rzmk}(x)=\ln(1+1/h_{\rzmk}),
\end{align}
which is an R-CBRF if the condition \eqref{eqn-18} holds. To achieve these two goals, we define the sliding surface function as
\begin{align*}
\mathbf{U}(x)&:=V_{\rzmk}(x)+\varrho B_{\rzmk}(x),
\end{align*}
where $\varrho>0$ is to ensure the condition \eqref{eqn-58}. If $x\rightarrow\partial\mathbb{C}$, then $h_{\rzmk}\rightarrow0$ and $B_{\rzmk}\rightarrow+\infty$, which thus shows the existence of $\varrho>0$. From Remark \ref{rmk-3}, the larger the value of $\varrho>0$ is, the lower the possibility of the reach of the state to $\partial\mathbb{C}$ is.

Let $M=1650, a_{0}=0.1, a_{1}=5, a_{2}=0.25, v_{\textsf{d}}=22, \varrho=50, \tau=0.2$ and $a\in[-2.5, 2.5]$. From the controller \eqref{eqn-50}, we consider the following two cases. The first case is $\mathbf{K}(x_{t})=5\sign(\mathbf{U}(x(t)))$, whereas the second case is $\mathbf{K}(x_{t})=2.2\mathbf{U}(x(t))-2\|\mathbf{U}(x_{t})\|^{2}/\mathbf{U}(x(t))$. For these two cases, the simulation results are shown in Figure \ref{fig-1}. From Figures (1b) and (1d), $h_{\rzmk}(x)>0$ is valid such that the safety objective is achieved. From Figures (1a) and (1c), the desired velocity of the follower car is achieved under the premise of the safety guarantee. Comparing with the delay-free case \citep{Ames2016control}, both stabilization and safety goals are achieved for the time-delay case, whereas the chattering phenomena exist in Figure \ref{fig-1} due to the dependence of the controller \eqref{eqn-50} on the sliding surface function. On the other hand, comparing Figures (1a) and (1c), the chattering phenomena are attenuated greatly in the second case (i.e., Figure (1c)) due to the controller \eqref{eqn-50} with \eqref{eqn-59}, where the time-delay trajectory is involved.

Since the goals are affected by the chattering phenomena, a practical way is to transform the first goal from a specified value to a bounded band, which is similar to the quasi-sliding mode band to limit the ultimate bound of the sliding surface function and to reduce the chattering phenomena \citep{Bartoszewicz2019discrete}. By taking this strategy, the desired velocity is given as the bands in Figures (2a) and (2c). In this case, Figure \ref{fig-2} shows the simulation results under different controllers. From Figures (2b) and (2d), we can see that the controller \eqref{eqn-50} with \eqref{eqn-59} leads the velocity of the follower car to reach a narrower band. In addition, for the controller \eqref{eqn-50} with \eqref{eqn-53}, the parameter $\mathrm{K}$ needs to be changed to 10 such that the velocity of the follower car to reach the band, which also shows the advantage of \eqref{eqn-59} over \eqref{eqn-53}.

\subsection{Master-Slave Synchronization}
\label{subsec-MSyn}

In this subsection, the master-slave synchronization problem is investigated, aiming to keep the master and slave robots synchronous while driving both to avoid the obstacle. This problem exists extensively in many industrial fields \citep{Rodriguez2004mutual}. Here, we consider the two-dimensional robot model, which is a simplification and transformation of the kinematic dynamics for robots \citep{Yalccin2001master}. The dynamics of the master robot is given as
\begin{align}
\label{eqn-65}
\dot{x}_{\master}(t)=A_{\master}x(t)+F_{\master}(x^{\dly}_{\master}, x^{\dly}_{\slave})+B_{\master}u_{\master},
\end{align}
where $x_{\master}\in\mathbb{R}^{2}$ is the position of the master robot, and $u_{\master}\in\mathbb{R}^{2}$ is the velocity of the master robot. $x^{\dly}_{\master}$ and $x^{\dly}_{\slave}$ are respectively time-delay states of the master and slave robots, and time delays come from the shared communication channel. Note that the time delay from the master robot to the slave robot is allowed to be different from the time delay from the slave robot to the master robot, and both time delays are bounded via the constants $\Delta_{\master}, \Delta_{\slave}\geq0$. $A_{\master}, B_{\master}\in\mathbb{R}^{2\times2}$ are constant matrices, and $F_{\master}: \mathcal{C}([-\Delta_{\master}, 0], \mathbb{R}^{2})\times\mathcal{C}([-\Delta_{\slave}, 0], \mathbb{R}^{2})\rightarrow\mathbb{R}^{2}$ is a continuously differentiable function showing the coupling between the master and slave robots due to the information communication \citep{Lee2006passive}. Similarly, the dynamics of the slave robot is
\begin{align}
\label{eqn-66}
\dot{x}_{\slave}(t)=A_{\slave}x_{\slave}(t)+F_{\slave}(x^{\dly}_{\master}, x^{\dly}_{\slave})+B_{\slave}u_{\slave},
\end{align}
where $x_{\slave}\in\mathbb{R}^{2}$ is the position of the slave robot, and $u_{\slave}\in\mathbb{R}^{2}$ is the velocity of the slave robot. $A_{\slave}, B_{\slave}\in\mathbb{R}^{2\times2}$ are constant matrices and $F_{\slave}: \mathcal{C}([-\Delta_{\master}, 0], \mathbb{R}^{2})\times\mathcal{C}([-\Delta_{\slave}, 0], \mathbb{R}^{2})\rightarrow\mathbb{R}^{2}$ is continuously differentiable.

\begin{figure*}
\centering
\begin{tabular}{cccc}
\includegraphics[width=0.48\columnwidth]{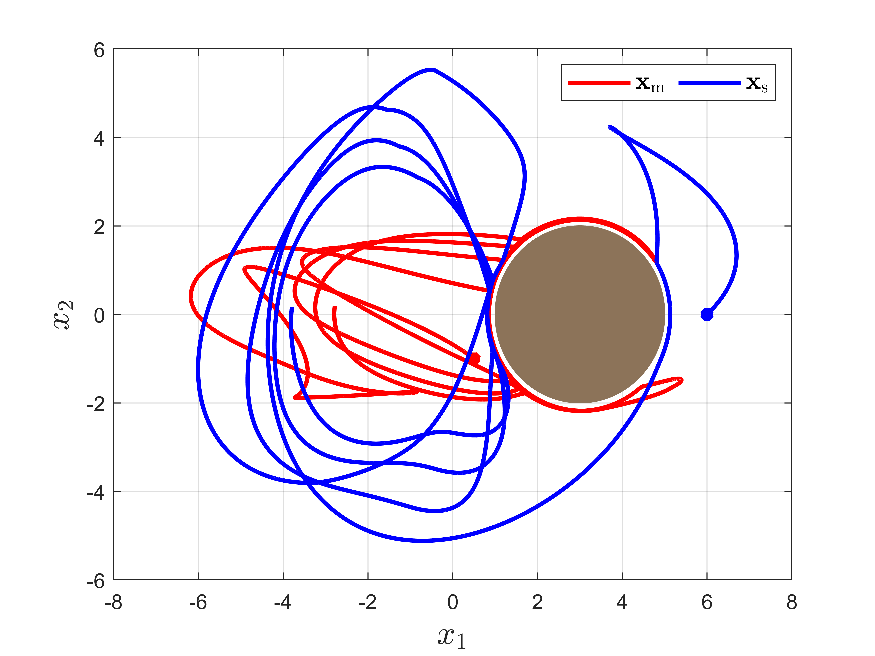}\vspace{-1pt} &
\includegraphics[width=0.48\columnwidth]{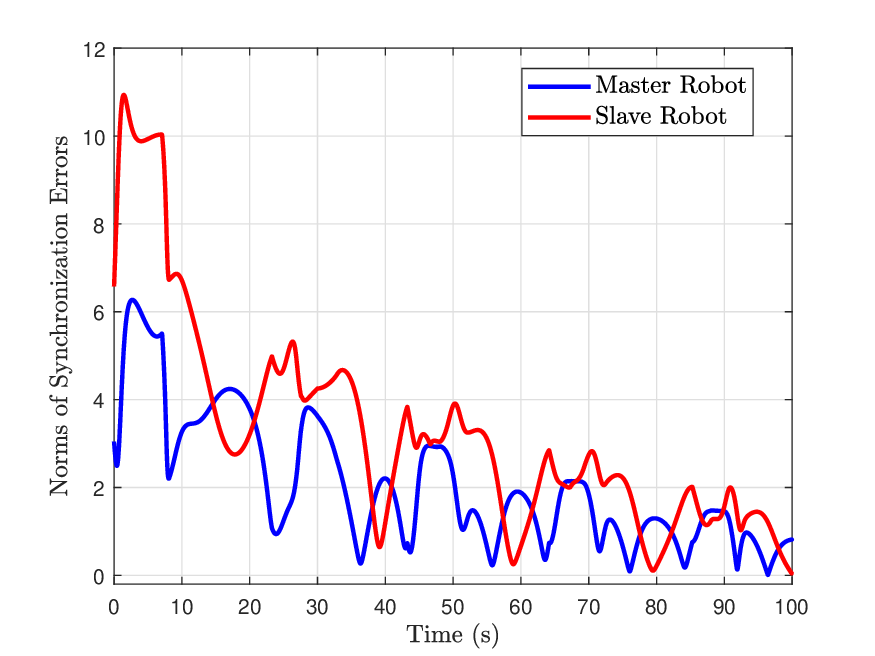}\vspace{-1pt} &
\includegraphics[width=0.48\columnwidth]{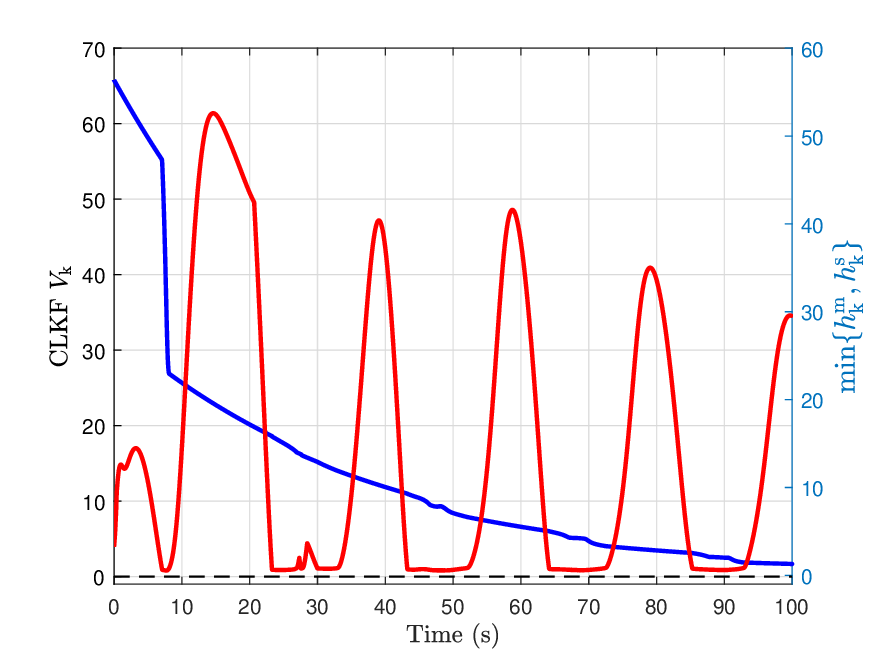}\vspace{-1pt}  &
\includegraphics[width=0.48\columnwidth]{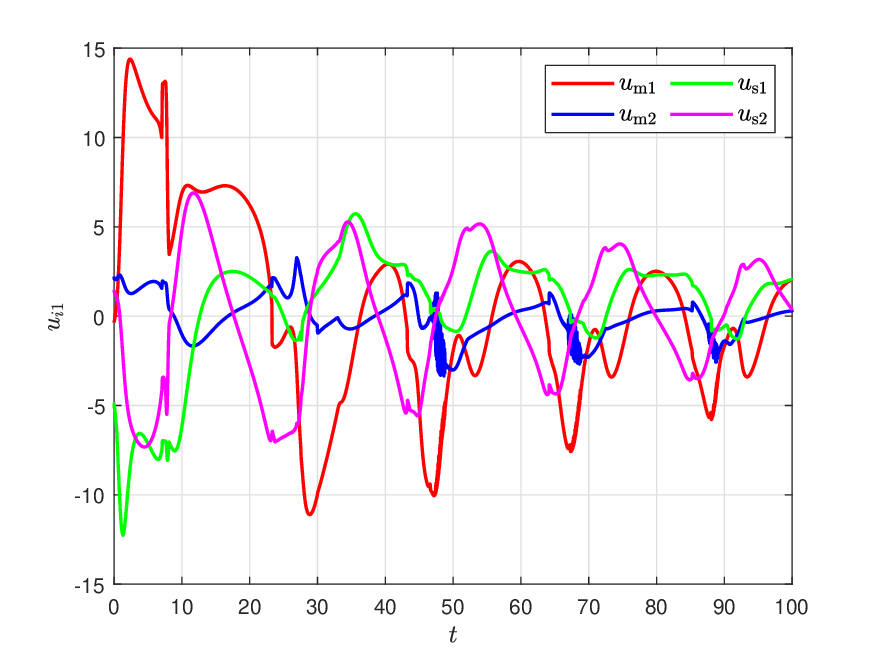} \vspace{-1pt} \\
\tiny{(a)} & \tiny{(b)} & \tiny{(c)} & \tiny{(d)}
\end{tabular}
\caption{\footnotesize Simulation of the master-slave synchronization problem under time-varying delays. (a) the state trajectories of the master and slave robots; (b) the norm evolution of the synchronization errors; (c) the evolution of the CLKF and CBKF; and (d) the control inputs for the master and slave robots.}
\label{fig-3}
\end{figure*}

To guarantee the synchronization objective, the desired reference trajectory for the master robot is given by $\mathbf{x}_{\master}\in\mathbb{R}^{2}$, and accordingly the reference trajectory for the slave robot is given by $\mathbf{x}_{\slave}\in\mathbb{R}^{2}$. To this end, the following Lyapunov-Krasovskii functional is introduced:
\begin{align}
\label{eqn-67}
V_{\ksvk}(x_{t})&:=e^{T}(t)P_{1}e(t)+\int^{t}_{t-\Delta_{\master}}e^{T}_{\master}(s)P_{2}e_{\master}(s)ds \nonumber \\
&\quad +\int^{t}_{t-\Delta_{\slave}}e^{T}_{\slave}(s)P_{3}e_{\slave}(s)ds,
\end{align}
where $P_{1}\in\mathbb{R}^{4\times4}, P_{2}, P_{3}\in\mathbb{R}^{2\times2}$ are positive-definite matrices, and $e(t):=(e_{\master}(t), e_{\slave}(t))$ is the synchronization error with $e_{\master}(t):=x_{\master}(t)-\mathbf{x}_{\master}(t)\in\mathbb{R}^{2}$ and $e_{\slave}(t):=x_{\slave}(t)-\mathbf{x}_{\slave}(t)\in\mathbb{R}^{2}$. From Definition \ref{def-8}, $V_{\ksvk}$ is a CLKF if the condition \eqref{eqn-31} holds.

To avoid the collision between the robots and the obstacle $\mathcal{O}$ (i.e., the brown region in Figure (3a)), we introduce the functional $h_{\ksvk}(\phi)=(\phi^2_{1}-3)^2+\phi^2_{2}-4$ with $\phi=(\phi_{1}, \phi_{2})$, and the safety objective is satisfied only when $h_{\ksvk}(\phi)\geq0$. With the function $h_{\ksvk}(\phi)$, the barrier-Krasovskii functional is given by
\begin{align}
\label{eqn-68}
\mathcal{B}_{\ksvk}(\phi)&:=-\ln(\exp(-\sgn(h^{\slave}_{\ksvk}(\phi_{\slave})-\varepsilon_{\slave})B^{\slave}_{\ksvk}(\phi_{\slave})) \nonumber \\
& \quad +\exp(-\sgn(h^{\master}_{\ksvk}(\phi_{\master})-\varepsilon_{\master})B^{\master}_{\ksvk}(\phi_{\master}))),
\end{align}
where $B^{\master}_{\ksvk}(\phi_{\master})=(1-\varepsilon_{\master}(h^{\master}_{\ksvk}(\phi_{\master}))^{-1})^{2}$, $B^{\slave}_{\ksvk}(\phi_{\slave})=(1-\varepsilon_{\slave}(h^{\slave}_{\ksvk}(\phi_{\slave}))^{-1})^{2}$, and $\varepsilon_{\master}, \varepsilon_{\slave}>0$ are constant to determine the warning regions for the master and slave robots. From \eqref{eqn-68}, the R-CBKF $\mathcal{B}_{\ksvk}$ is applied and has effects on the control design only when the robots move into the warning regions. $\mathcal{B}_{\ksvk}$ is smoothly separable since we can rewrite it as $\mathcal{B}_{\ksvk}(\phi(0))+(\mathcal{B}_{\ksvk}-\mathcal{B}_{\ksvk}(\phi(0)))$, and is an R-CBKF if item (ii) in Definition \ref{def-10} holds. To achieve the synchronization and obstacle avoidance simultaneously, we define the sliding surface functional as
\begin{align*}
\mathbf{U}(\phi)&:=V_{\ksvk}(\phi)+\mathcal{B}_{\ksvk}(\phi).
\end{align*}

Let $A_{\master}=\diag\{1, 0.3\}, A_{\slave}=\diag\{0.7, 0.6\}, B_{\master}=B_{\slave}=I, F_{\master}=\diag\{1, 0.3\}x^{\dly}_{\master}-0.5x^{\dly}_{\slave}$ and $F_{\slave}=\diag\{0.5, 0.8\}x^{\dly}_{\slave}-0.3x^{\dly}_{\master}$. The delays are bounded by $\Delta_{\master}=0.5$ and $\Delta_{\slave}=0.2$. The reference trajectories are $\mathbf{x}_{\master}:=(2\sin(0.3t),$ $2\cos(0.3t))$ and $\mathbf{x}_{\slave}:=\mathbf{x}_{\master}-(1, 0)$. That is, the slave robot is expected to keep a constant distance with the master robot. Let $P_{1}=I, P_{2}=P_{3}=I$ in \eqref{eqn-67}, and $\varepsilon_{\master}=\varepsilon_{\slave}=0.84$ in \eqref{eqn-68}, which implies that the warning region is $\{\phi\in\mathcal{C}([-\Delta_{\master}, 0], \mathbb{R})\times\mathcal{C}([-\Delta_{\slave}, 0], \mathbb{R}): (\phi^2_{1}-3)^2+\phi^2_{2}\leq2.2^2\}$. Using the controller \eqref{eqn-50} with $\mathbf{K}(\phi)=0.025\mathbf{U}(\phi)$, the simulation results are shown in Figure \ref{fig-3}. From Figures (3a) and (3c), $\min\{h^{\master}_{\ksvk}(\phi_{\master}), h^{\slave}_{\ksvk}(\phi_{\slave})\}>0$ and the obstacle avoidance is achieved. From Figures (3b) and (3c), the synchronization errors and the CLKF $V_{\ksvk}$ decrease with the increase of time, which implies the satisfaction of the synchronization objective. The control inputs for the master and slave robots are presented in Figure (3d). Finally, if we reduce the CBKF \eqref{eqn-68} to the CBRF $\mathcal{B}_{\ksvk}(\phi(0))$, similar simulation results can be obtained such that the synchronization and obstacle avoidance are achieved, which shows the potential combination between the CBRF and CLKF for time-delay systems.

\section{Conclusion}
\label{sec-conclusion}

This paper provided a novel framework for the control design of safety-critical systems with state delays. Following both Razumikhin and Krasovskii approaches, we proposed Razumikhin-type control functions and Krasovskii-type control functionals to investigate the stabilization and safety objectives. In particular, we established the small control property for the time-delay case, and further extended the classic Sontag's formula into the time-delay case. To achieve the stabilization and safety objectives simultaneously, the proposed control functions/functionals were combined via the sliding surface function, which ensures the existence and continuity of the derived controller. The derived results were demonstrated via the application to multi-agent systems. Future work will be devoted to distributed safety-critical control of cyber-physical systems with time delays. Based on control Lyapunov and barrier functions/functionals, the optimal control for time-delay systems is an interesting future topic.

%

\appendix
\section{Technical Lemma}
\label{sec-appendix}

\setcounter{equation}{0}
\renewcommand{\theequation}{A.\arabic{equation}}

\begin{lemma}
\label{lem-A1}
Given two continuously differentiable functions $\mathbf{X}, \mathbf{Y}: [-\Delta, \infty)\rightarrow\mathbb{R}$, and for all $t>0$,
\begin{align}
\label{eqn-A1}
D^{+}\mathbf{X}(t)&\leq-\alpha(\mathbf{X}(t))+\varepsilon_{1}\beta\left(\sup_{\theta\in[-\Delta, 0]}\mathbf{X}(t+\theta)\right), \\
\label{eqn-A2}
D^{+}\mathbf{Y}(t)&\geq-\alpha(\mathbf{Y}(t))+\varepsilon_{2}\beta\left(\sup_{\theta\in[-\Delta, 0]}\mathbf{Y}(t+\theta)\right),
\end{align}
where $\varepsilon_{1}, \varepsilon_{2}\in[0, 1]$ with $\varepsilon_{2}>\varepsilon_{1}$, and $\alpha, \beta: \mathbb{R}\rightarrow\mathbb{R}$ are continuously nondecreasing functions. If $\mathbf{X}(\theta)=\mathbf{Y}(\theta)$ for $\theta\in[-\Delta, 0]$, then $\mathbf{X}(t)\leq\mathbf{Y}(t)$ for all $t\in\mathbb{R}^{+}$.
\end{lemma}

\begin{proof}
We prove that $\mathbf{X}(t)\leq\mathbf{Y}(t)$ via \emph{reductio ad absurdum}. Obviously, $\mathbf{X}(t)\leq\mathbf{Y}(t)$ for $t=0$.

Suppose that $\mathbf{X}(t)\leq\mathbf{Y}(t)$ for all $t>0$. If not, then there exists at least one $t^{\ast}\in\mathbb{R}^{+}$ such that $\mathbf{X}(t^{\ast})>\mathbf{Y}(t^{\ast})$. Define $t^{\prime}:=\inf\{t\in\mathbb{R}^{+}: \mathbf{X}(t)>\mathbf{Y}(t)\}$. Therefore, we have
\begin{align}
\label{eqn-A3}
\mathbf{X}(t)&\leq\mathbf{Y}(t), \quad  t\in(0, t^{\prime});\\
\label{eqn-A4}
\mathbf{X}(t^{\prime})&=\mathbf{Y}(t^{\prime}); \\
\label{eqn-A5}
\mathbf{X}(t)&>\mathbf{Y}(t), \quad t\in(t^{\prime}, t^{\prime}+\Delta t),
\end{align}
where $\Delta t>0$ is arbitrarily small. From \eqref{eqn-A4}-\eqref{eqn-A5}, we have that for any $\tau\in(0, \Delta t)$,
\begin{align*}
\frac{\mathbf{X}(t^{\prime}+\tau)-\mathbf{X}(t^{\prime})}{\tau}&>\frac{\mathbf{Y}(t^{\prime}+\tau)-\mathbf{Y}(t^{\prime})}{\tau},
\end{align*}
which implies from the definition of the Dini derivative that $D^{+}\mathbf{X}(t^{\prime})\geq D^{+}\mathbf{Y}(t^{\prime})$. On the other hand, from \eqref{eqn-A3}-\eqref{eqn-A4}, $\sup_{\theta\in[-\Delta, 0]}\mathbf{X}(t^{\prime}+\theta)\leq\sup_{\theta\in[-\Delta, 0]}\mathbf{Y}(t^{\prime}+\theta)$ for all $\theta\in[-\Delta, 0]$, and further from \eqref{eqn-A1}-\eqref{eqn-A2}, $D^{+}\mathbf{X}(t^{\prime})<D^{+}\mathbf{Y}(t^{\prime})$, which results in a contradiction. Therefore, we conclude that $\mathbf{X}(t)\leq\mathbf{Y}(t)$ for all $t\in(0, \infty)$.
\hfill $\blacksquare$
\end{proof}


\bibliographystyle{elsarticle-num-names}

\end{document}